\documentclass[aps,prl,twocolumn,groupedaddress,amsmath,superscriptaddress,english]{revtex4}
\usepackage{amsfonts}
\usepackage{amssymb}
\usepackage{amsmath}
\usepackage[dvips]{graphicx}
\usepackage{xcolor}
\usepackage{dcolumn}
\usepackage{bm}
\usepackage{soul}

\newcommand{\red}{\textcolor{red}}

\definecolor{colorone}{rgb}{1,0.36,0.03}
\definecolor{colortwo}{rgb}{0.54,0.71,0.03}
\definecolor{colorthree}{rgb}{0.01,0.51,0.93}
\definecolor{colorfour}{rgb}{0.47,0.26,0.58}

\def\a{\alpha}

\makeatletter
\newcommand*{\rom}[1]{\expandafter\@slowromancap\romannumeral #1@}
\makeatother

\usepackage{hyperref}
\hypersetup{colorlinks=true,citecolor=blue,linkcolor=blue,filecolor=blue,urlcolor=blue,breaklinks=true,linktocpage=true}

\usepackage{enumitem}
\usepackage{amsmath}
\usepackage{graphicx}
\usepackage{amsfonts}
\usepackage{amssymb}
\usepackage{hyperref}
\usepackage{color}
\usepackage{mathrsfs}
\usepackage{isomath}
\usepackage{amsthm}
\usepackage{epstopdf}
\usepackage{txfonts}
\usepackage{float}

\newcounter{remark}
\newenvironment{remark}[1][]{\refstepcounter{remark}\par\medskip\noindent%
\textbf{Remark~\theremark #1} }{\medskip}

\usepackage{tikz}
\usetikzlibrary{patterns}
\makeatletter
\def\grd@save@target#1{%
  \def\grd@target{#1}}
\def\grd@save@start#1{%
  \def\grd@start{#1}}
\tikzset{
  grid with coordinates/.style={
    to path={%
      \pgfextra{%
        \edef\grd@@target{(\tikztotarget)}%
        \tikz@scan@one@point\grd@save@target\grd@@target\relax
        \edef\grd@@start{(\tikztostart)}%
        \tikz@scan@one@point\grd@save@start\grd@@start\relax
        \draw[minor help lines,magenta] (\tikztostart) grid (\tikztotarget);
        \draw[major help lines] (\tikztostart) grid (\tikztotarget);
        \grd@start
        \pgfmathsetmacro{\grd@xa}{\the\pgf@x/1cm}
        \pgfmathsetmacro{\grd@ya}{\the\pgf@y/1cm}
        \grd@target
        \pgfmathsetmacro{\grd@xb}{\the\pgf@x/1cm}
        \pgfmathsetmacro{\grd@yb}{\the\pgf@y/1cm}
        \pgfmathsetmacro{\grd@xc}{\grd@xa + \pgfkeysvalueof{/tikz/grid with coordinates/major step}}
        \pgfmathsetmacro{\grd@yc}{\grd@ya + \pgfkeysvalueof{/tikz/grid with coordinates/major step}}
        \foreach \x in {\grd@xa,\grd@xc,...,\grd@xb}
        \node[anchor=north] at (\x,\grd@ya) {\pgfmathprintnumber{\x}};
        \foreach \y in {\grd@ya,\grd@yc,...,\grd@yb}
        \node[anchor=east] at (\grd@xa,\y) {\pgfmathprintnumber{\y}};
      }
    }
  },
  minor help lines/.style={
    help lines,
    step=\pgfkeysvalueof{/tikz/grid with coordinates/minor step}
  },
  major help lines/.style={
    help lines,
    line width=\pgfkeysvalueof{/tikz/grid with coordinates/major line width},
    step=\pgfkeysvalueof{/tikz/grid with coordinates/major step}
  },
  grid with coordinates/.cd,
  minor step/.initial=.2,
  major step/.initial=1,
  major line width/.initial=2pt,
}
\makeatother

\newcommand{\ket}[1]{|#1\rangle}
\newcommand{\bra}[1]{\langle#1|}

\newcommand{\Tr}{\mbox{Tr}}

\def\be{\begin{equation}}
\def\ee{\end{equation}}
\def\ba{\begin{array}}
\def\ea{\end{array}}
\def\Tr{\mathrm{Tr}}

\newcommand{\rr}{\mathbf{r}}
\newcommand{\p}{\mathbf{p}}
\newcommand{\q}{\mathbf{q}}

\newcommand{\s}{\mathbf{s}}
\newcommand{\ttt}{\mathbf{t}}
\newcommand{\x}{\mathbf{x}}
\newcommand{\y}{\mathbf{y}}

\makeatletter
\newcommand*\bigcdot{\mathpalette\bigcdot@{.5}}
\newcommand*\bigcdot@[2]{\mathbin{\vcenter{\hbox{\scalebox{#2}{$\m@th#1\bullet$}}}}}
\makeatother

\usepackage{booktabs}
\usepackage{adjustbox}
\newcommand{\BigWedge}{\mathord{\adjustbox{valign=C,totalheight=.6\baselineskip}{$\boldsymbol\bigwedge$}}}
\newcommand{\BigVee}{\mathord{\adjustbox{valign=C,totalheight=.6\baselineskip}{$\boldsymbol\bigvee$}}}
\newcommand{\BigCup}{\mathord{\adjustbox{valign=C,totalheight=.7\baselineskip}{$\bigcup$}}}

\newtheorem{thm}{Theorem}
\newtheorem{lem}{Lemma}

\newtheorem{proposition}{Proposition}
\newtheorem{definition}{Definition}

\begin{document}

\title{$\mbox{The Complementary Information Principle of Quantum Mechanics}$}
\author{Yunlong~Xiao}
\email{mathxiao123@gmail.com}
\address{Department of Mathematics and Statistics and Institute for Quantum Science and Technology, University of Calgary, Calgary, Alberta T2N 1N4, Canada}
\author{Kun~Fang}
\email{kf383@cam.ac.uk}
\address{Department of Applied Mathematics and Theoretical Physics, University of Cambridge, Cambridge, CB3 0WA, UK}
\author{Gilad~Gour}
\email{gour@ucalgary.ca}
\address{Department of Mathematics and Statistics and Institute for Quantum Science and Technology, University of Calgary, Calgary, Alberta T2N 1N4, Canada}

\begin{abstract}
The uncertainty principle bounds the uncertainties about incompatible measurements, clearly setting quantum theory apart from the classical world. Its mathematical formulation via uncertainty relations, plays an irreplaceable role in quantum technologies. However, neither the uncertainty principle nor uncertainty relations can fully describe the complementarity between quantum measurements. As an attempt to advance the efforts of complementarity in quantum theories, we formally propose a \emph{complementary information principle}, significantly extending the one introduced by Heisenberg. 
First, we build a framework of black box testing consisting of pre- and post-testing with two incompatible measurements,
introducing a rigorous mathematical expression of complementarity with definite information causality. Second, we provide majorization lower and upper bounds for the complementary information by utilizing the tool of semidefinite programming. In particular, we prove that our bounds are optimal under majorization due to the completeness of the majorization lattice. Finally, as applications to our framework, we present a general method to outer-approximating all uncertainty regions and also establish fundamental limits for all qualified joint uncertainties. 

\end{abstract}

\maketitle

\textbf{\textit{Introduction.---}} 
Quantum mechanics has revolutionized our understanding of the physical world, leaving us a key message that our world is inherently unpredictable.
The well-known \emph{Heisenberg's uncertainty principle}~\cite{Heisenberg1927} states that it is impossible to prepare a state such that its outcome probability distributions from two incompatible measurements are both sharp. This heuristic idea was first formulated by Kennard~\cite{Kennard1927} (see also the work of Weyl \cite{Weyl1927}) into a mathematically rigorous inequality for position and momentum measurements, 
where the uncertainties were quantified through the standard deviation \cite{Robertson1929,Schrodinger1930,Huang2012,Maccone2014,Xiao2016W,Xiao2017I,Guise2018}. Later on, an alternative quantitative formulation based on entropic measures has been introduced by Bia\l ynicki-Birula and Mycielski \cite{Bialynicki1975}. Recently, physicists have also employed majorization to study uncertainty relations \cite{Partovi2011,Gour2013,Rudnicki2013,Rudnicki2014,Rudnicki2018,Yuan2019E}. All of these are known as preparation uncertainty with detailed reviews given in \cite{Busch2007,Wehner2010,Coles2017}.

The uncertainty principle manifests another deeper concept in quantum theory, namely \emph{complementarity}~\cite{Bohr1928}---a given physical attribute can only be revealed at the price of another complementary attribute being suppressed. As a more general concept, complementarity can also be exhibited through ``duality paradoxes'', such as wave-particle duality~\cite{Coles2014E}. The arise of complementarity differentiates quantum theory from its classical counterpart, leading to a plethora of applications such as 
entanglement detection \cite{Horodecki1994,Giovannetti2004,Guhne2004,Guhne2009}, Einstein-Podolsky-Rosen (EPR) steering detection \cite{Xiao2018Q,Rutkowski2017,Riccardi2018,Costa2018E,Costa2018S,Uola2019} as well as quantum key distribution \cite{Hall1995,Cerf2002,Renner2005,Renes2009}.
 
In recent years, uncertainty relations have been used as a standard tool to explore the complementarity nature of quantum measurements. A series of efforts have been devoted to seeking the optimal bounds on uncertainty relations for given specific uncertainty measures, such as Shannon or R\'enyi entropies \cite{Deutsch1983,Partovi1983,Kraus1987,Maassen1988,Ivanovic1992,Sanchez1993,Ballester2007,Wu2009,Berta2010,Li2011,Prevedel2011,Huang2011,Tomamichel2011,Coles2012,Coles2014,Kaniewski2014,Furrer2014,Berta2016,Xiao2016S, Xiao2016QM, Xiao2016U,Xiao2018H,Chen2018,Coles2019,Li2019}. However, most of this enormous body of work focuses on a single entropy measure, whereas the complete information about the physical attributes of incompatible measurements is fully contained in the outcome probability vectors themselves~\cite{Saha2019}. Hence, inevitable losses of information occurs whenever a high dimensional probability vector is projected into its one-dimensional entropic value. To have a full-scale understanding of complementarity between incompatible measurements and to be able to find more practical applications, it is therefore essential to consider a more general framework, in which the full probability vectors are considered.

In this paper, we introduce a new principle which we call the {\it complementary information principle} (CIP), characterizing the information trade-off between incompatible measurements, as characterized with respect to their full outcome probability vectors. This principle is different than the uncertainty principle as it does not rely on an uncertainty quantifier such as entropy or standard deviation. Since the CIP involves pairs of probability vectors (corresponding to the two incompatible measurements), it introduces a new type of partial order between pairs of probability vectors that we call marginal majorization. Based on this order and with the help of semidefinite programming (SDP), we investigate qualitative and quantitative aspects of our CIP, providing insights into the fundamental limits of quantum measurements. As an application to our CIP, we provide tight outer approximation of uncertainty regions with any given uncertainty measures in general \cite{Narasimhachar2016}, improving and generalizing the bounds given in \cite{Abdelkhalek2015}
which are only restricted to specific uncertainty measures, and in~\cite{Li2015,Abbott2016,Schwonnek2017,Busch2019} which are limited to weak measures.
The generality and efficiency of this approach makes it suitable as a benchmark for the forthcoming research on uncertainty regions. Finally, we discuss another application of our framework in bounding general forms of joint uncertainties. 
Indeed, our results outperform the universal uncertainty relations (UUR) \cite{Gour2013} and direct-sum majorization uncertainty relations (DSMUR) \cite{Rudnicki2014,Rudnicki2018} since we make an efficient use of the information gain from the pre-testing. We also showcase that our uncertainty regions are more informative than the best known Maassen and Uffink's (MU) entropic uncertainty relations \cite{Maassen1988}, such as breaking  the limitation of the harmonic condition. We stress that our principle is universal and captures  the essence of complementarity in quantum mechanics as it does not rely on any specific form of uncertainty measures.

 The main results of our study are presented as Theorems, with all the technical details of the proofs delegated to the Supplemental Material~\cite{SM}.
 
\begin{figure}[h]
\centering
\begin{tikzpicture}

  
  \node[] at (-2.7,-0.4) {(a) PUR: $\mathcal{J}(\p,\q) \geqslant c$.};

  \draw[very thick,dashed] (-5,3) rectangle (-0.75,0); 
  \draw[very thick,fill=black!70,black!70] (-4.8,2.5) rectangle (-3.8,1.8) node[pos=0.5,white] {\large $\Gamma_\rho$}; 
  \draw[very thick] (-3,2.5) rectangle (-2,1.8); 

  \draw[very thick,->] (-2.5,1.95) -- (-2.35,2.4);
  \draw[very thick] (-2.15,1.95) arc (10:170:0.35);

  \draw[very thick,fill=black!70,black!70] (-4.8,1.2) rectangle (-3.8,0.5) node[pos=0.5,white] {\large $\Gamma_\rho$}; 
  \draw[very thick,->] (-2.5,0.65) -- (-2.35,1.1);

  \draw[very thick] (-2.15,0.65) arc (10:170:0.35);
  \draw[very thick] (-3,1.2) rectangle (-2,0.5); 

  \node[] at (-2.5,2.7) {$M$};
  \node[] at (-2.5,0.3) {$N$};

  \draw[very thick,->] (-3.8,2.15) -- (-3,2.15);
  \draw[very thick,->] (-3.8,0.8) -- (-3,0.8); 
  
  \draw[very thick,->] (-2,2.15) -- (-1.15,2.15) -- (-1.15,1.75);
  \draw[very thick,->] (-2,0.8) -- (-1.15,0.8) -- (-1.15,1.2); 
  \node[] at (-1.6,2.4) {$\p$};
  \node[] at (-1.6,0.55) {$\q$};
  \node[] at (-1.3,1.475) {$\mathcal{J}(\p,\q)$};

  \node[] at (-3.4,2.4) {$\mathcal H$};
  \node[] at (-3.4,0.55) {$\mathcal H$};

  \begin{scope}[shift={(-0.25,0)}]
   \draw[fill=magenta,opacity=0.2] (-0.2,3) -- (3.7,3) -- (3.7,1.5) -- (-0.2,1.5);
  \draw[fill=cyan,opacity=0.2] (-0.2,1.5) -- (3.7,1.5) -- (3.7,0) -- (-0.2,0);

  \draw[very thick,dashed] (-0.2,3) rectangle (3.7,0); 
  \draw[very thick,fill=black!70,black!70] (0,2.5) rectangle (1,1.8) node[pos=0.5,white] {\large $\Gamma_\rho$}; 
  \draw[very thick] (1.8,2.5) rectangle (2.8,1.8); 

  \draw[very thick,->] (2.3,1.95) -- (2.45,2.4);
  \draw[very thick] (2.65,1.95) arc (10:170:0.35);

  \draw[very thick,fill=red!70,black!70] (0,1.2) rectangle (1,0.5) node[pos=0.5,white] {\large $\Gamma_\rho$}; 
  \draw[very thick] (1.8,1.2) rectangle (2.8,0.5);

  \draw[very thick,->] (2.3,0.65) -- (2.45,1.1);
  \draw[very thick] (2.65,0.65) arc (10:170:0.35);

  \draw[very thick,->] (1,2.15) -- (1.8,2.15) node[pos=0.5,shift={(0,0.22)}] {$\mathcal H$}; 
  \draw[very thick,->] (1,0.8) -- (1.8,0.8)  
         node[pos=0.5,shift={(0,-0.22)}] {$\mathcal H$};


  \draw[thick, dashed] (-0.2,1.5) -- (3.7,1.5);

  \node[] at (0.55,2.75) {\footnotesize \textbf{Pre-testing}};
  \node[] at (0.6,0.2) {\footnotesize \textbf{Post-testing}};

  \node[] at (2.3,2.7) {$M$};
  \node[] at (2.3,0.3) {$N$};
  
  \draw[very thick,->] (2.8,2.15) -- (3.3,2.15);
  \node[] at (3.48,2.15) {$\p$};
  \draw[very thick,->] (2.8,0.8) -- (3.3,0.8);
   \node[] at (3.48,0.8) {$\q$};

  \node[] at (2,-0.4) {(b) CIR: $\rr \prec \q \prec \ttt$.};
  \end{scope}
\end{tikzpicture}
\caption{(color online) Scenarios of preparational uncertainty relation (PUR) (a) and complementary information relation (CIR) (b). PUR gives a state-independent bound, namely $c$, and conditioned on information gain $\p$, CIR provides the optimal bounds for $\q$ under majorization ``$\prec$''; that is $\rr \prec \q \prec \ttt$.}
\label{bb}
\end{figure}
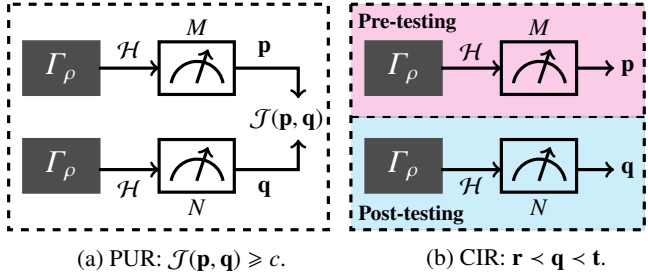

\textbf{\textit{Framework of Black Box Testing.---}} 
The basic task of black box testing is shown in Fig.~\ref{bb}. An unknown black box prepares two independent and identically distributed resources $\rho$, which are being tested by incompatible measurements $M$ and $N$ chronologically. Without loss of generality, we assume that the test with $M$ (a.k.a. \emph{pre-testing}) is performed first. After that, we do the test with $N$ (a.k.a. \emph{post-testing}). 
Each test is performed repeatedly, returning us an outcome probability distribution.
Prior to the test, the knowledge associated with the outcomes is associated with the ``uncertainty'' of the quantum measurement. Once the test is completed and its outcome is physically observed, this uncertainty turns into an ``information gain''. In this paper we focus on the complementarity between the information gain from the pre-testing and the uncertainty associated with post-testing (before the quantum measurement is performed). 

Formally, consider a preparation channel $\mathrm{\Gamma}_{\rho}$ inside the black box generating quantum state $\rho$ \cite{Kraus1983}.
Then the outcome probability distribution of the pre-testing with (basis) measurement $M= \{ |u_{j}\rangle\langle u_j| \}_{j=1}^n$ is given by the Born's rule
%
$\p :=(\langle u_{j}|\,\rho\,|u_{j}\rangle)_{j=1}^n$.
Since the pre-testing is repeated with the same measurement, $M$, rather than a complete tomography process, there exist many density matrices leading to the same probability distribution of the measurement outcomes. If the pre-testing outcome occurs according to the probability distribution $\p = (c_j)_{j=1}^n$, we collect all candidates of $\rho$ compatible with this outcome as
\begin{align}
S(M,\p):=\{\,\rho \geqslant 0\;|\; \langle u_{j}|\,\rho\,|u_{j}\rangle = c_j, \forall j = 1,\cdots,n\,\}.
\end{align} 
This set of quantum states forms the information gain from the pre-testing, manifesting the physical attribute of $M$ via its classical outcome $\p$, and narrowing down the possible range of $\rho$ being tested. Since the post-testing will be performed over the same quantum state $\rho$, its outcome is necessarily confined by the information gain.  For a fixed post-testing measurement $N=\{ |v_{\ell}\rangle\langle v_\ell| \}_{\ell=1}^n$, we denote the set of all possible probability vectors $\q :=(\langle v_{\ell}|\,\rho\,|v_{\ell}\rangle)_{\ell=1}^n$ by 
\begin{align}
\mathcal{Q}(N|M,\p):=\left\{\big(\langle v_{\ell}|\,\rho\,|v_{\ell}\rangle\big)_{\ell=1}^n ~|~ \rho \in S(M,\p) \right\}.
\end{align}
We will also abbreviate $\mathcal{Q}(N|M,\p)$ as $\mathcal{Q}$ for simplicity. 
 Since the tuple $\{\p,\mathcal Q\}$ captures fully the complementarity between the measurements $M$ and $N$, we will work with this set directly, and quantify the uncertainty of $\mathcal Q$ based on the information gain $\p$.

\vspace{0.2cm}
\textbf{\textit{Complementary Information Principle.---}} 
Following the observation given in~\cite{Gour2013}, the uncertainty associated with a probability vector cannot be larger than the uncertainty associated with its randomly relabelled version. That is, a probability vector $\x \in \mathbb{R}^n$ is more uncertain than $\y \in \mathbb{R}^n$ if and only if $\x$ is majorized  by $\y$, $\x \prec \y$, i.e. $\sum_{j=1}^{k} x^{\downarrow}_{j} \leqslant \sum_{j=1}^{k} y^{\downarrow}_{j}$ for all $1\leqslant k \leqslant n-1$ \cite{Majorization}. Here the down-arrow indicates that the components of the corresponding vectors are arranged in a non-increasing order. 
Since we compare probability vectors in $\mathcal Q$ with the same information gain $\p$, we define the \emph{right-majorization} as $\left( \p_{1} ,  \q_{1} \right) \prec_{R}\left( \p_{2} ,  \q_{2} \right)$ if $\p_{1} = \p_{2}$ and  $\q_{1} \prec \q_{2}$,
emphasizing that two sets of the black box testing admit the same pre-testing distribution, while the first post-testing outcome is more uncertain than the second one. 
The \emph{left-majorization} can be similarly defined as $\left( \p_{1} ,  \q_{1} \right) \prec_{L}\left( \p_{2} ,  \q_{2} \right)$ if $\q_{1} = \q_{2}$ and  $\p_{1} \prec \p_{2}$. Both $\prec_{R}$ and $\prec_{L}$ are called here \emph{marginal majorizations}, which is different from the relative majorization discussed in \cite{Renes2016}.

Based on the notion of marginal majorizations, we are now in a position to state our main results.  Due to the fact that a probability simplex with majorization forms a complete lattice \cite{Rapat1991,Bondar1994,Bosyk2019}, then any possible probability vector $\q \in \mathcal{Q}$ must be confined within two unique probability vectors $\rr$, $\ttt$. The following theorems establish an explicit construction of the optimal choices of $\rr$ and $\ttt$.

\begin{thm}\label{main thm1}
Let $M= \{ |u_{j}\rangle\langle u_j| \}_{j=1}^n$ and $N=\left\{ |v_{\ell}\rangle\langle v_\ell| \right\}_{\ell=1}^n$ be the measurements of pre- and post-testing respectively. If the outcome probability distribution from $M$ is given by $\p = (c_j)_{j=1}^n$, then any outcome probability $\q$ from $N$ is bounded as $\left( \p ,  \rr \right) \prec_{R} \left( \p ,  \q \right) \prec_{R} \left( \p ,  \s \right)$ with
\begin{align}
\rr &= \left(R_{1}, R_{2}, \ldots, R_{n} \right) :=\left(r_{1}, r_{2}-r_{1}, \ldots, r_{n}-r_{n-1} \right),\label{eq: mj bound r}\\
\s &= \left(S_{1}, S_{2}, \ldots, S_{n} \right) :=\left(s_{1}, s_{2}-s_{1}, \ldots, s_{n}-s_{n-1} \right).
\end{align}
Each elements $r_{k}$, $s_{k}$ can be efficiently computed via the following SDPs respectively, 
\begin{align}
r_{k} & = \min \Big\{x \; \big|\; x \geqslant \Tr \left( N_{\mathcal{I}_k}\,\rho \right), \;\forall \mathcal{I}_{k}\subset [n],\; \rho \in S(M,\p)\Big\}\\
s_{k} & = \max \Big\{\Tr \left( N_{\mathcal{I}_k}\,\rho \right)\;\big|\;\forall \mathcal{I}_{k}\subset [n],\; \rho \in S(M,\p)\Big\},
\end{align}
where $\mathcal{I}_k$ denotes a subset of $[n]\equiv\{1,\cdots,n\}$ with cardinality $k$ and $N_{\mathcal{I}_k}:= \sum_{\ell \in \mathcal{I}_{k}} | v_\ell \rangle\langle v_\ell |$ is a partial sum with index $\mathcal I_k$.
\end{thm}

An intuitive understanding of this result can be illustrated in terms of Lorenz curves \cite{Majorization}. For any probability vector $\x=(x_i)_{i=1}^n$ in an non-increasing order, its associated Lorenz curve $\mathcal{L}(\x)$ is defined as the linear interpolation of the points $\big\{\big(k, \sum_{i=1}^k x_{i}\big)\big\}_{k=0}^{n}$ with the convention $(0,0)$ for $k=0$. Then the majorization relation $\x\prec \y$ can be geometrically interpreted as $\mathcal L(\x)$ laying everywhere below $\mathcal L(\y)$. Accordingly, to find the optimal $\rr$ and $\ttt$ such that $\rr \prec \q \prec \ttt$ for all $\q \in \mathcal Q$ is equivalent to find the tightest Lorenz curves bounding $\mathcal{L}(\q)$ from below and above for all $\q \in \mathcal Q$, respectively. For this purpose, the curves $\left\{\left(k, r_{k}\right)\right\}_{k=0}^{n}$ and $\left\{\left(k, s_{k}\right)\right\}_{k=0}^{n}$ with $r_0 = s_0 = 0$ from Theorem~\ref{main thm1} are  taken as the minimum and maximum of all the Lorenz curves  generating from the set $\mathcal{Q}$. Note that any Lorenz curve is concave by definition, and the minimum of a set of concave functions remains concave. Therefore, $\left\{\left(k, r_{k}\right)\right\}_{k=0}^{n}$ gives us a well-defined Lorenz curve, indicating the optimality of the lower bound $\rr$. However, this is not the case for $\left\{\left(k, s_{k}\right)\right\}_{k=0}^{n}$ since it is not necessarily concave in general. An explicit example is given in the Supplemental Material~\cite{SM}.

To obtain the optimal upper bound $\ttt$, it suffices to perform an additional \emph{flatness process} \cite{Cicalese2002}, constructing the tightest Lorenz curve from $\left\{\left(k, s_{k}\right)\right\}_{k=0}^{n}$.
Specifically, let $j$ be the smallest integer in $\left\{2, \ldots, n\right\}$ such that $S_{j}>S_{j-1}$, and $i$ be the greatest integer in $\left\{1, \ldots, j-1\right\}$ such that $S_{i-1} \geqslant (\sum_{k=i}^{j} S_{k})/(j-i+1):=a$. Define the vector $\bm t$ as
\begin{align}\label{eq; mj bound t}
\ttt := \left(T_{1}, \ldots, T_{n}\right) \quad \text{with} \quad T_{k} = 
     \begin{cases}
       a & \text{for}\quad k = i, \ldots, j \\
       S_{k} & \text{otherwise.} \\ 
     \end{cases}
\end{align}
Finally, the bounds $\rr$ in~\eqref{eq: mj bound r} and $\ttt$ in~\eqref{eq; mj bound t} can be ensured as optimal for the set $\mathcal Q$ under the order of majorization.

\begin{thm}\label{main thm2}
Based on the same settings as in Theorem~\ref{main thm1}, for any probability vectors $\x$ and $\y$ such that $\x \prec \q \prec \y$ for all $\q \in \mathcal{Q}$, it holds $\x \prec \rr \prec \q \prec \ttt \prec \y$.
\end{thm}

Since majorization only forms a partial order, the optimal bounds $\rr$ and $\ttt$ not necessarily belong to $\mathcal Q$ in general. That is, these bounds may not be achieved by any quantum state $\rho \in S(M,\p)$ except for the qubit case where majorization is a total order.  Moreover, the flatness process is required in general except for the qubit and qutrit cases where the linear interpolation of $\left\{\left(k, s_{k}\right)\right\}_{k=0}^{n}$ is already concave. In terms of the computational complexity of our results, we need to solve one SDP of size $n\times n$ with $n=\dim(\rho)$ for every lower bound element $r_{k}$ and solve $\binom{n}{k}$ SDPs of size $n \times n$ for every upper bound element $s_{k}$. Since each SDPs are independent for different index $\mathcal I_k$, the element $s_{k}$ can still be solved efficiently via parallel computations even for large dimension $n$. Moreover, we would like to mention that by applying the flatness process directly, we could tighter the bound of UUR and obtain the optimal bound of DSMUR.

An explicit example of the above theorems is given in Fig.~\ref{ill} where we choose the pre-testing with measurement $M = \{|0\rangle\langle 0|, |1\rangle\langle 1|, |2\rangle\langle 2|\}$ with outcome probability distribution $\p = (1/6,1/6,2/3)$ and the post-testing with measurement $N$ given by $|v_0\rangle=(1/2,0,\sqrt{3}/2)$, $|v_1\rangle=(0,1,0)$ and $|v_2\rangle=(-\sqrt{3}/2,0,1/2)$. Each solid color line is the Lorenz curve of a probability vector $\q \in \mathcal Q$. The optimal lower bound $\rr$ and the optimal upper bound $\ttt$ are depicted as thick black lines. These two boundaries fully characterize the uncertainty of the set $\mathcal Q$ based on the information gain from the pre-testing.

Our results inspire a new form of ``information causality''~\cite{Pawlowski2009}: the information that an observer can gain from a state in the past confines the uncertainty associated with the same state in the future. 
In other words, the physical attribute of the uncertainty associated with $N$ can only be exhibited at the expense of the information gained from $M$. {One extreme case is that without any information gain from the pre-testing, we can only infer that $S(M,\p)$ is the set of all quantum states, and hence obtaining the trivial bounds for the post-testing, i.e. $(1/n,\ldots,1/n)\prec\q\prec(1,0,\ldots,0)$. Another extreme case is the Heisenberg's uncertainty principle, which corresponds to the situation where we obtain an outcome from pre-testing with certainty. That is, the outcome probability vector $\p$ has one entry equal to $1$ and $0$ otherwise. This implies that the information gain $S(M,\p)=\{ |u_{j}\rangle\langle u_{j}| \}_{j=1}^n$.  Direct calculations gives the first element $T_{1}$ in upper bound $\ttt$ equal to the maximal overlap between measurements $M$ and $N$, namely $\max_{jk}| \langle u_{j} | v_{k} \rangle |^{2}$. Hence, we have $\q \prec \ttt \neq (1,0,\ldots, 0)$, indicating a non-zero uncertainty of the measurement $N$, whenever the $M$ and $N$ do not have any common eigenvectors.

\begin{figure}[t]
\centering
\begin{tikzpicture}
  \node[inner sep=0pt] at (0,0) {\includegraphics[width=6cm]{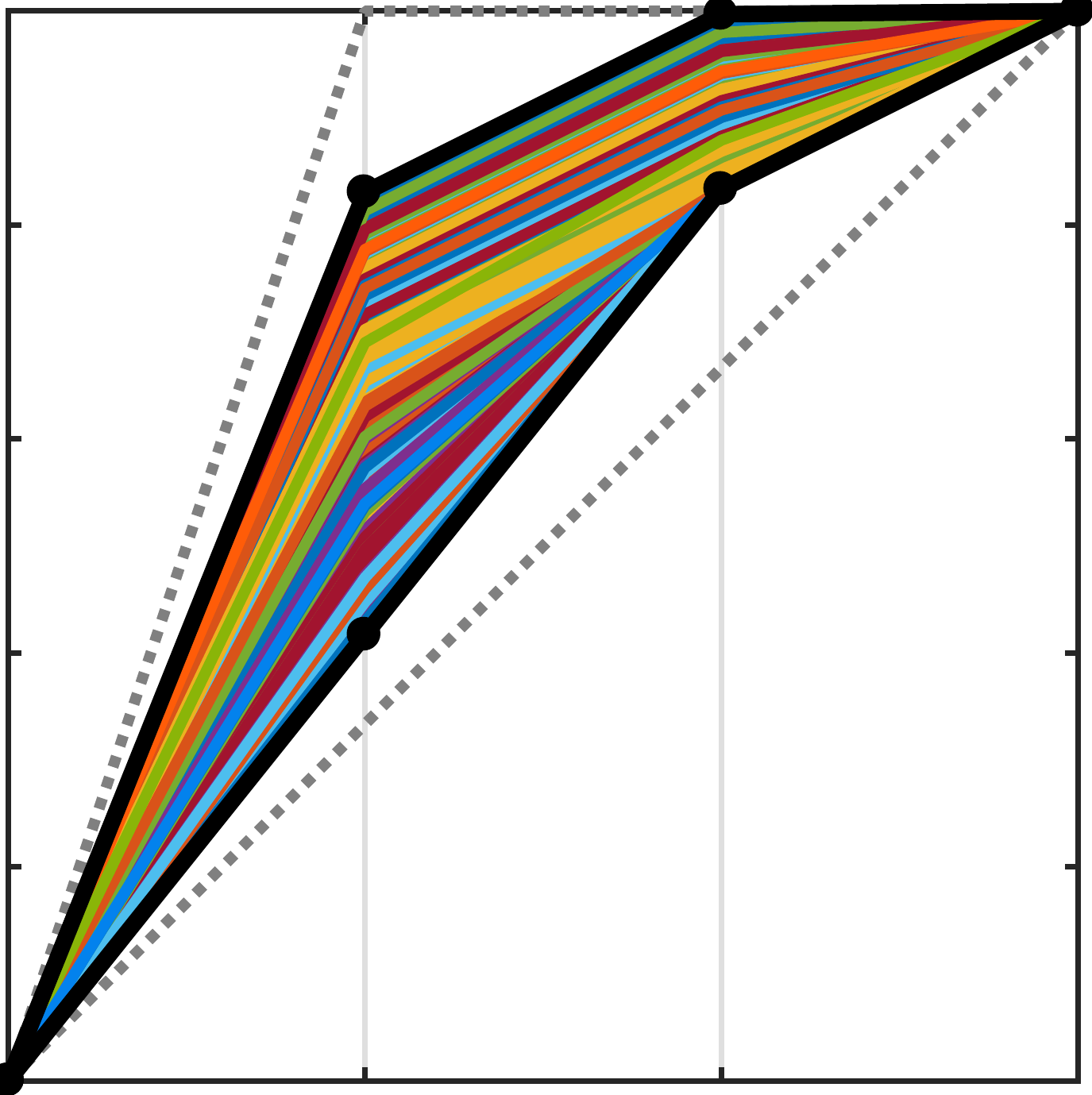}};

  \node at (-3.3,-3) {\footnotesize $0$};
  \node at (-3.3,-1.75) {\footnotesize $0.2$};
  \node at (-3.3,-0.58) {\footnotesize $0.4$};
  \node at (-3.3,0.6) {\footnotesize $0.6$};
  \node at (-3.3,1.78) {\footnotesize $0.8$};
  \node at (-3.3,2.9) {\footnotesize $1$};
  
  \node at (-1,-3.2) {\footnotesize $1$};
  \node at (1,-3.2) {\footnotesize $2$};
  \node at (3.,-3.2) {\footnotesize $3$};

  \draw[->, thick] (-1,-0.5) -- (-0.15,-1.2) node[right,shift={(-0.1,-0.05)}] {\footnotesize $r_{1}$}; 
  \draw[->, thick] (1,1.95) -- (2.05,1.2) node[right,shift={(-0.1,-0.05)}] {\footnotesize $r_{2}$}; 

  \draw[->, thick] (-1,2) -- (-2.3,1) node[left,shift={(0.1,0)}] {\footnotesize $t_{1}$}; 
  \draw[->, thick] (1,2.95) -- (-2.2,2.6) node[left,shift={(0.1,0)}] {\footnotesize $t_{2}$};

  \draw[dashed,->,line width=0.4mm] (-0.8,-0.28) -- (0.25,-1.15) node[right,shift={(-0.1,-0.05)}] {\footnotesize $\rr$}; 
  \draw[dashed,->,line width=0.4mm,magenta] (0,1.4) -- (1.45,0.2) node[right,shift={(-0.1,-0.05)}] {\footnotesize $\q \in \mathcal{Q}$};
   \draw[dashed,->,line width=0.4mm,gray] (0,0) -- (1,-0.8) node[right,shift={(-0.1,-0.05)}] {\footnotesize $(1/3,1/3,1/3)$};  
   \draw[dashed,->,line width=0.4mm,gray] (-2,0) -- (-0.15,-1.6) node[right,shift={(-0.1,-0.05)}] {\footnotesize $(1,0,0)$};  
   \draw[dashed,->,line width=0.4mm] (0,2.5) -- (-2,2) node[right,shift={(-0.3,0)}] {\footnotesize $\ttt$}; 
   \filldraw [black] (-0.8,-0.28) circle [radius=1.5pt];
   \filldraw [magenta] (0,1.4) circle [radius=1.5pt];
   \filldraw [gray] (0,0) circle [radius=1.5pt];
   \filldraw [gray] (-2,0) circle [radius=1.5pt];
   \filldraw [black] (0,2.5) circle [radius=1.5pt]; 

\end{tikzpicture}
\caption{(color online) An example of Theorem~\ref{main thm1} and \ref{main thm2}. The dotted line are the Lorenz curves of trivial bounds $(1/3,1/3,1/3)$, $(1,0,0)$. The thick black lines are Lorenz curves of the optimal lower and upper bounds given by $\rr$ and $\ttt$, respectively. Each solid color line is the  Lorenz curve of a probability vector $\q$ from $\mathcal{Q}$.}
\label{ill}
\end{figure}
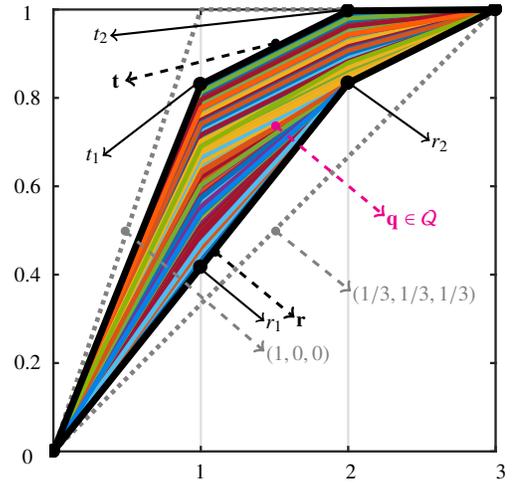

\textbf{\textit{Universal Uncertainty Regions.---}}
Since the pioneering work of Deutsch \cite{Deutsch1983}, much has been done in the direction of lower-bounding the joint uncertainties $f \left( \p(\rho) \right) + g \left( \q(\rho) \right) \geqslant b$, where $\p(\rho)$ and $\q(\rho)$ are outcome probability distributions of the pre- and post-testing on state $\rho$ respectively, and $f$, $g$ are valid uncertainty measures \cite{Partovi1983,Kraus1987,Maassen1988,Ivanovic1992,Sanchez1993,Ballester2007,Wu2009,Berta2010,Li2011,Prevedel2011,Huang2011,Tomamichel2011,Coles2012,Coles2014,Kaniewski2014,Furrer2014,Berta2016,Xiao2016S, Xiao2016QM, Xiao2016U,Xiao2018H,Chen2018,Coles2019,Li2019}, i.e., non-negative Schur-concave functions, including Shannon entropy, R\'enyi entropy, elementary symmetric functions and so forth. In the case of \emph{uncertainty region} 
\begin{align}
\mathcal R(f,g):= \underset{\rho}{\BigCup}\, \big\{\left( f \left( \p(\rho) \right), g \left(\q(\rho) \right) \right)\big\},
\end{align}
the relation $f \left( \p(\rho) \right) + g \left( \q(\rho) \right)$ is nothing but a straight line with slope $-1$ in the coordinate plane of $(f(\p),g(\q))$, and its optimal lower bound $b = \min_{\rho} \left\{f \left( \p(\rho) \right) + g \left( \q(\rho) \right)\right\}$ is then achieved at the tangent line to the bottom left of $\mathcal R(f,g)$ as shown in Fig.~\ref{geo}. Namely the description of uncertainty regions \red{is} much more informative than uncertainty relations. However, no efficient method is known to characterize the region $\mathcal{R} \left( f , g \right)$ in general.
As an application of the majorization bounds in Theorem~\ref{main thm1} and \ref{main thm2}, we can provide a general approach for tight outer-approximations. 

Consider the statistics set $\mathcal R := {\BigCup}_\rho\, \big\{ (\p(\rho),\q(\rho))\big\}$ in $\mathbb R^n \times \mathbb R^n$ by collecting all compatible pairs of pre- and post-testing outcomes. Note that the set of all quantum states $\rho$ can be divided into equivalent classes $\{S(M,\p):\p\in \mathsf S_n\}$ based on the outcome probability distribution $\p$ from the pre-testing, where $\mathsf S_n := \{\,\x \in \mathbb R^n: \sum_i x_i = 1, x_i \geq 0, \forall i\,\}$ is the probability simplex of dimension $n$. Then $\mathcal R$ can be fine-grained as 
\begin{align}
  \mathcal R = \underset{\p \in \mathsf S_n}{\BigCup} \mathcal R_\p,\quad \text{with}\quad \mathcal R_{\p}:= \big\{\,(\p,\q): \q \in Q(M,N,\p)\,\big\}.
\end{align} 
For any fixed $\p$, the majorization bounds $\rr$, $\ttt$ set a boundary for the set $\mathcal Q(M,N,\p)$. Thus we have the relaxation 
\begin{align}
  \mathcal R_{\p} \subseteq \widetilde{\mathcal R}_{\p},\quad \text{with}\quad \widetilde{\mathcal R}_{\p}:= \big\{\,(\p,\q): \rr \prec \q\prec \ttt\,\big\}.
\end{align}
As a consequence, by taking the union with respect to $\p$, we have the relaxation of the whole region,
\begin{align}
  \mathcal R \subseteq \widetilde{\mathcal R}  \quad \text{with}\quad \widetilde{\mathcal R}:= \underset{\p \in \mathsf S_n}{\BigCup} \widetilde{\mathcal R}_\p.
\end{align} 
Finally, any uncertainty region $\mathcal R(f,g)$ can be retrieved by projecting $\mathcal R$ from $\mathbb R^n \times \mathbb R^n$ to $\mathbb R \times \mathbb R$ via the uncertainty measures $f$, $g$, and the projection of $\widetilde{\mathcal R}$ will give us an outer-approximation of  $\mathcal R(f,g)$ accordingly. Since $\widetilde{\mathcal R}$ can be used to generate an approximation of uncertainty region with any measures, we thus name it a \emph{universal uncertainty region}.  We emphasize that this universal region $\widetilde{\mathcal R}$ as well as its projections can be explictly depicted by running $\p$ over the probability simplex, which is significantly more tractable than characterizing $\mathcal{R}$ by taking $\rho$ over the set of all quantum states \cite{Geometry}.

\begin{thm}
Let $M= \{ |u_{j}\rangle\langle u_j| \}_{j=1}^n$ and $N=\left\{ |v_{\ell}\rangle\langle v_\ell| \right\}_{\ell=1}^n$ be the measurements of pre- and post-testing respectively. For any uncertainty measures $f$ for $M$ and $g$ for $N$, their uncertainty region is outer-approximated as $\mathcal{R} \left( f , g \right) \subseteq \widetilde{\mathcal{R}} \left( f , g \right)$
with $\widetilde{\mathcal{R}} \left( f , g \right):= \big\{\,(f(\p),g(\q)): (\p,\q) \in \widetilde{\mathcal R}\,\big\}$. In particular, the outer-approximation is optimal $\mathcal{R} \left( f , g \right) = \widetilde{\mathcal{R}} \left( f , g \right)$ when $n=2$.
\end{thm}

Note that for any given pre-testing outcome $\p$, the majorization bounds $\rr$ and $\ttt$ can be explicitly computed. Combining the Schur-concavity of the uncertainty measure $g$, the fine-grained outer-approximation can be simplified as 
\begin{align}
\widetilde{\mathcal{R}}_{\p} \left( f , g \right) = 
\left\{ \left( f \left( \p \right) , y \right) ~|~  g ( \ttt ) \leqslant y \leqslant g ( \rr ) \right\}.
\end{align} 
By running $\p$ over the probability simplex, we can explicitly depict the whole region $\widetilde{\mathcal R}(f,g) = \BigCup_{\p\in\mathsf S_n} \widetilde{\mathcal{R}}_{\p}(f,g)$. A schematic diagram is given in Fig. \ref{geo}, which legibly explains how our approximation method works. It is also worth mentioning that by swapping the role of pre- and post-testing, we can get another outer-approximation $\mathcal R \subseteq \widetilde{\mathcal W}$. Taking the intersection $\widetilde{\mathcal R} \cap \widetilde{\mathcal W}$ will lead to a potentially tighter result $\mathcal R(f,g) \subseteq \widetilde{\mathcal R}(f,g) \cap \widetilde{\mathcal W}(f,g)$.

The extension of our results to multiple measurements is also straightforward. For instance, let us consider a black box testing process consisting of one pre-testing with $M$ and $m$ post-testings with $N_{i}$ ($i \in \left\{1, \ldots, m\right\}$). If $M$ indicates a probability vector $\p$, then the outcome probability distribution $\q_{i}$ of $N_{i}$ are bounded by
$\rr_{i} \prec \q_{i} \prec \ttt_{i}$,
where $\rr_{i}$ and $\ttt_{i}$ are obtained from Theorem~\ref{main thm1} and \ref{main thm2} by replacing $N$ with $N_{i}$. With this strategy, we can delineate an outer-approximation for the high-dimensional uncertainty region ${\BigCup}_\rho \big\{\left( f \left( \p(\rho) \right),g_1 (\q_1(\rho)), \cdots, g_{m} \left( \q_{m}(\rho) \right) \right)\big\}$ with any uncertainty measures $f$ and $g_1,\cdots,g_m$.

Due to the generality of our approach, the approximation is not guaranteed to work well for every uncertainty measures. But we should stress that our approximation can be computed explicitly and is valid for any eligible uncertainty measures, liberating us from specific forms of uncertainty relations. {More importantly,  this is the first efficient method to approximating uncertainty regions in general, which can be used as a benchmark for future works.}

\begin{figure}[t]
\centering
\begin{tikzpicture}
\draw[very thick,->] (-3,-3) -- (3.2,-3);
\draw[very thick,->] (-3,-3) -- (-3,3.2);
\draw[very thick] (-3,3) -- (3,3);
\draw[very thick] (3,-3) -- (3,3);

\draw[line width=0.4mm,magenta] (-3,-0.5) -- (-3,1);
\draw[line width=0.4mm,magenta] (-3,1) .. controls (-2.65,1.8) and (-2.95,2.5) .. (-1,3);
\draw[line width=0.4mm,magenta] (-1,3) -- (2,3);
\draw[line width=0.4mm,magenta] (2,3) .. controls (3,2.4) and (2.8,2.5) .. (3,1.6);
\draw[line width=0.4mm,magenta] (3,1.6) -- (3,-1);
\draw[line width=0.4mm,magenta] (3,-1) .. controls (2,-2.9) and (2.5,-2.2) .. (1.5,-3);
\draw[line width=0.4mm,magenta] (1.5,-3) -- (-1,-3);
\draw[line width=0.4mm,magenta] (-1,-3) .. controls (-1.75,-2.675) and (-2.35,-2.65) .. (-3,-0.5);
\draw[fill=magenta,opacity=0.2] (-3,-0.5) -- (-3,1) -- (-2.75,1.8) -- (-2.65,2) -- (-2.45,2.3) -- (-2.3,2.45) -- (-2,2.62) -- (-1.75,2.75) -- (-1,3) -- (2,3) -- (2.75,2.5) -- (2.85,2.3) -- (3,1.6) -- (3,-1) -- (2.65,-1.7) -- (2.32,-2.3) -- (2.29,-2.39) -- (1.5,-3) -- (-1,-3) -- (-1.6,-2.75) -- (-2,-2.5) -- (-2.335,-2.1) -- (-2.62,-1.5) -- (-2.75,-1.2) -- (-3,-0.5) -- (-3,-0.5);

\draw[line width=0.4mm,magenta] (-1.5,-3) -- (-3,-1.5);

\draw[line width=0.3mm,dashed,cyan] (-3,-1) -- (-3,1.5);
\draw[line width=0.3mm,dashed,cyan] (-3,1.5) .. controls (-2.55,2.4) and (-2.95,2.7) .. (-1.5,3);
\draw[line width=0.3mm,dashed,cyan] (-1.5,3) -- (2.2,3);
\draw[line width=0.3mm,dashed,cyan] (2.2,3) .. controls (3,2.6) and (2.8,2.6) .. (3,2);
\draw[line width=0.3mm,dashed,cyan] (3,2) -- (3,-1.3);
\draw[line width=0.3mm,dashed,cyan] (3,-1.3) .. controls (2.5,-2.7) and (2.5,-2.2) .. (1.8,-3);
\draw[line width=0.3mm,dashed,cyan] (1.8,-3) -- (-0.8,-3);
\draw[line width=0.3mm,dashed,cyan](-0.8,-3) .. controls (-2,-2.975) and (-2.35,-2.65) .. (-3,-1);
\draw[fill=cyan,opacity=0.2] (-3,1) -- (-2.75,1.8) -- (-2.65,2) -- (-2.45,2.3) -- (-2.3,2.45) -- (-2,2.62) -- (-1.75,2.75) -- (-1,3) -- (-1.53,3) -- (-2.2,2.81) -- (-2.6,2.55) -- (-3,1.45) -- (-3,1);
\draw[fill=cyan,opacity=0.2] (2,3) -- (2.75,2.5) -- (2.85,2.3) -- (3,1.6) -- (3,2) -- (2.9,2.45) -- (2.8,2.63) -- (2.7,2.71) -- (2.5,2.84) -- (2.2,3) -- (2,3);
\draw[fill=cyan,opacity=0.2] (3,-1) -- (2.65,-1.7) -- (2.32,-2.3) -- (2.29,-2.39) -- (1.5,-3) -- (1.8,-3) -- (2,-2.78) -- (2.5,-2.36) -- (2.7,-2.05) -- (2.85,-1.7) -- (2.92,-1.5) -- (3,-1.3) -- (3,-1);
\draw[fill=cyan,opacity=0.2] (-1,-3) -- (-1.6,-2.75) -- (-2,-2.5) -- (-2.335,-2.1) -- (-2.62,-1.5) -- (-2.75,-1.2) -- (-3,-0.5) -- (-3,-1) -- (-2.8,-1.47) -- (-2.5,-2.1) -- (-2.1,-2.65) -- (-1.5,-2.94) -- (-1,-3);

\draw[line width=0.3mm,cyan] (-1.75,-3) -- (-3,-1.75);

\draw[very thick,orange] (-2.62,2.5) -- (2.88,2.5);

 \node at (-3,3.4) {\footnotesize $f(\p)$};
 \node at (3.5,-3) {\footnotesize $g(\q)$};
 \node at (-3.2,-3) {\footnotesize $0$};
 
 \filldraw [gray] (-1.5,-3) circle [radius=1.5pt];
 \filldraw [gray] (-3,-1.5) circle [radius=1.5pt];
 \node at (-3.2,-1.45) {\footnotesize $b$};
 \node at (-1.5,-3.3) {\footnotesize $\min\limits_{\rho}\{f(\p)+g(\q)\}(:=b)$};
 
 \draw[->, thick] (-2.45,-2.05) -- (-2.15,-1.75) node[right,shift={(-0.1,0)}] {\footnotesize $f(\p)+g(\q)\geqslant b$};
 
  \filldraw [gray] (-2.62,2.5) circle [radius=1.5pt];
   \draw[->, thick] (-2.62,2.5) -- (-2.15,2) node[right,shift={(-0.18,-0.15)}] {\footnotesize $g(\rr)$};
  \filldraw [gray] (2.88,2.5) circle [radius=1.5pt];
  \draw[->, thick] (2.88,2.5) -- (2.7,1.2) node[right,shift={(-0.4,-0.17)}] {\footnotesize $g(\ttt)$};
   \draw[->, thick] (-2.3,2.8) -- (-1,1.5) node[right,shift={(0,0)}] {\footnotesize outer-approximation $\widetilde{\mathcal{R}} \left( f , g \right)$};
    \draw[->, thick] (-1.5,2.5) -- (-1.2,2.2) node[right,shift={(0.1,-0.2)}] {\footnotesize outer-approximation $\widetilde{\mathcal{R}}_{\p} \left( f , g \right)$};
    \node at (-0.3,2.2) {\footnotesize fine-grained};
    \draw[->, thick] (-2.65,2) -- (-1.8,1.15) node[right,shift={(-0.15,-0.2)}] {\footnotesize uncertainty region $\mathcal{R} \left( f , g \right)$};
  
\end{tikzpicture}
\caption{(color online) A schematic diagram depicts the uncertainty region $\mathcal{R} \left( f , g \right)$ (magenta) with its outer-approximation $\widetilde{\mathcal{R}} \left( f , g \right)$ (cyan), and the fined-grained outer-approximation $\widetilde{\mathcal{R}}_{\p} \left( f , g \right)$ (orange). The optimal uncertainty relations $f(\p)+g(\q)\geqslant \min_{\rho}\{f(\p)+g(\q)\}$ is tangent to the left lower boundaries of $\mathcal{R} \left( f , g \right)$.}
\label{geo}
\end{figure}
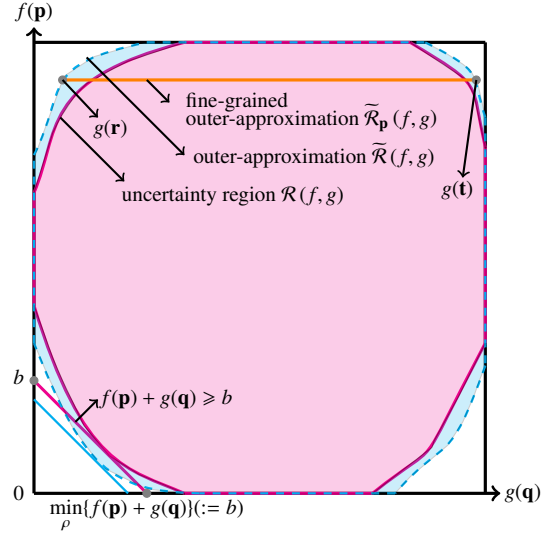

Detailed comparisons of our results with previously known results, such as MU entropic uncertainty relations \cite{Maassen1988}, UUR \cite{Gour2013}, and DSMUR \cite{Rudnicki2014,Rudnicki2018}, are given in the Supplemental Material \cite{SM}. In terms of the qubit case, we find that MU bound is tight only when one of the orders of the R\'enyi entropies is infinite, i.e., $(\alpha,\beta)=(1/2,\infty)$ or $(\infty,1/2)$. Otherwise, MU bound can be improved. We also give an explicit example showing that our result can outperform the bound given in \cite{Rudnicki2018} despite the spectrum of density operator is utilized when constructing their bound. Compared to the universal uncertainty relations in \cite{Gour2013,Rudnicki2013,Rudnicki2014,Rudnicki2018} which can be used to provide lower bounds for $f\left(\p\otimes\q\right)$ and $f\left(\p\oplus\q\right)$, our universal uncertainty region works more generally for any arbitrary qualified measure of joint uncertainties.

\textbf{\textit{Fundamental Limits for Joint Uncertainties.---}}
As another illustration of the generality of our framework, we study the joint uncertainties given by the most general measure $\mathcal{J}: \mathbb R^n \times \mathbb R^n \rightarrow \mathbb R$ for a pair of probability vectors $(\p,\q) \in \mathcal R$~\cite{Narasimhachar2016}. Such a measure includes the usual forms $f(\p)+g(\q)$, $f(\p)g(\q)$, $f\left(\p\otimes\q\right)$ and $f\left(\p\oplus\q\right)$ as special cases. To capture the essential properties of a measure of joint uncertainties, it has been argued in~\cite{Narasimhachar2016} that $\mathcal{J}$ should meet the following postulates~: (\romannumeral 1) Non-negativity: $\mathcal{J}(\p,\q)\geqslant 0$; (\romannumeral 2) Monotonicity under randomly relabelling: $\mathcal{J}(D_{1}\p,D_{2}\q) \geqslant \mathcal{J}(\p,\q)$ for all doubly stochastic matrices $D_{1}$ and $D_{2}$. The characterization of the joint uncertainties $\mathcal{J}(\p,\q)$ is crucial in the study of quantum information and quantum measurements, leading to a plethora of applications~\cite{Horodecki1994,Giovannetti2004,Guhne2004,Guhne2009,Xiao2018Q,Rutkowski2017,Riccardi2018,Costa2018E,Costa2018S,Uola2019}. In particular, any state-independent lower bound $b$ of $\mathcal{J}(\p,\q)$ leads to a uncertainty relation $\mathcal{J}(\p,\q)\geqslant b$ while any state-independent upper bound $a$ of $\mathcal{J}(\p,\q)$ gives us a reverse uncertainty relation $a\geqslant\mathcal{J}(\p,\q)$. A natural question is to ask how to find $a$ and $b$ for joint uncertainties $\mathcal{J}(\p,\q)$ in general. 
Let us now apply our results to provide an answer. 

First, by associating \emph{Hardy-Littlewood-P\'olya theorem} \cite{Hardy1929}, which states that two probability vectors $\x \prec \y$ if and only if $\x = D \y$ for some doubly stochastic matrix $D$, with our right-majorization relation $\left( \p ,  \rr \right) \prec_{R} \left( \p ,  \q \right) \prec_{R} \left( \p ,  \ttt \right)$, we obtain $\mathcal{J}(\p,\rr) \geqslant \mathcal{J}(\p,\q) \geqslant \mathcal{J}(\p,\ttt)$ for each $\p$. Again let $\p$ run over all probability simplex, the joint uncertainties must be confined as $\max_{\p \in \mathsf S_n}\mathcal{J}(\p,\rr) \geqslant \mathcal{J}(\p,\q) \geqslant \min_{\p \in \mathsf S_n}\mathcal{J}(\p,\ttt)$. Swapping the role of pre- and post-testings, we can drive a left-majorization relation $\left( \bm{u} ,  \q \right) \prec_{L} \left( \p ,  \q \right) \prec_{L} \left( \bm{v} ,  \q \right)$ for each $\q$. Then it follows that $\max_{\q \in \mathsf S_n}\mathcal{J}(\bm{u},\q) \geqslant \mathcal{J}(\p,\q) \geqslant \min_{\q \in \mathsf S_n}\mathcal{J}(\bm{v},\q)$. Finally, taking the intersections, we have the state-independent bounds $a\geqslant\mathcal{J}(\p,\q)\geqslant b$ with
\begin{align}
&a:=\min\left\{\max_{\p \in \mathsf S_n}\mathcal{J}(\p,\rr), \max_{\q \in \mathsf S_n}\mathcal{J}(\bm{u},\q)\right\},\notag\\
&b:=\max\left\{\min_{\p \in \mathsf S_n}\mathcal{J}(\p,\ttt), \min_{\q \in \mathsf S_n}\mathcal{J}(\bm{v},\q)\right\}.
\end{align}
In particular, the bounds $a$ and $b$ are tight for the qubit case. More remarkably, our method works for all quantified joint uncertainty measures and provides strong supports for finding their fundamental limits. 

\textbf{\textit{Conclusions.---}}
We have proposed a new information principle, fully characterizing the complementarity between quantum measurements. Compared with the standard uncertainty principle, the CIP does not depend on a single quantifier of uncertainty, but instead contain all the information in the outcome probability distributions. Hence, the CIP is more informative and captures the essence of complementarity between measurements in quantum physics. This principle provides us an effective method in determining the boundary of the uncertainty region with any non-negative Schur-concave functions and establishes the limits for all joint uncertainties. Our method can, in principle, be applied to positive operator valued measures (POVMs) without modification.

Conventional studies on general quantum resource theories (QRTs) \cite{Chitambar2019} consider quantum state transformation with the assumption that we have full knowledge of the density operator of the resource state, for which a complete tomography is required \cite{Sugiyama2013}. In our Supplemental Material \cite{SM}, we further discuss an application of our result in a more practical scenario where only partial knowledge of the resource state is needed.

The connection between the results presented here and other areas, including entanglement \cite{Horodecki1994,Giovannetti2004,Guhne2004,Guhne2009} and EPR steering detection \cite{Xiao2018Q,Rutkowski2017,Riccardi2018,Costa2018E,Costa2018S,Uola2019}, moderate deviation of majorization-based resource interconversion \cite{Chubb2019} and quantum key distribution \cite{Cerf2002,Renner2005,Renes2009}, are particularly important in the context of quantum information. We left these directions for future work. The experimental demonstration of this paper, performed in a photonic system, is currently in progress. The experimental data and our theoretical results fit well \cite{Yuan2019}.

\textbf{\textit{Acknowledgments.---}}We would like to thank Xin Wang, Yuan Yuan, Guo-Yong Xiang, Nengkun Yu, Jie Xie, Lijian Zhang, Zhihao Ma, Gaurav Saxena, Carlo Maria Scandolo, and Barry C. Sanders for fruitful discussions. Y. X., and G. G. acknowledge financial support from the Natural Sciences and Engineering Research Council of Canada (NSERC). K.F. was supported by the University of Cambridge Isaac Newton Trust Early Career grant RG74916.

\clearpage

\onecolumngrid
\begin{center}
\vspace*{.5\baselineskip}
{\textbf{\large Supplemental Material: \\[3pt] Complementary Information Principle and Universal Uncertainty Regions Regions}}\\[1pt] \quad \\
\end{center}

\setcounter{equation}{0}
\setcounter{figure}{0}
\setcounter{table}{0}
\setcounter{section}{0}
\setcounter{thm}{0}

This supplemental material provides a more detailed analysis and proofs of the results in the main text. We may reiterate some of the steps in the main text to make the supplemental material more explicit and self-contained.

\section{Proof of Theorem 1}

Before proceeding, it is worth noting that our majorization inequalities $\rr \prec \q \prec \bm{s}$ is based on the following restricted conditions
\begin{align}
\langle u_{j}|\,\rho\,|u_{j}\rangle = c_{j}, \quad\forall j \in \left\{1, 2, \ldots, n\right\}.
\end{align}
This means we are considering the possible range (under majorization) of probability vector $\q$, generated from performing the quantum measurement $N$ on a state $\rho$ with $\rho \in S(M,\p):=\{\,\rho \geqslant 0\ |\, \langle u_{j}|\,\rho\,|u_{j}\rangle = c_j, \forall j\,\}$. 

The following lemma gives a useful characterization of the partial sum of the $k$ largest values in a probability vector.
\begin{lem}\label{majorization lemma}
  For any probability vector $\p$ with non-increasing order $(p_j^{\downarrow})_{j=1}^n$, we have the partial sum
  $\sum_{j=1}^{k} p^{\downarrow}_{j} = \max_{\mathcal I_k} \sum_{s\in \mathcal I_k}p_s$ where the maximization on the r.h.s. is taken over all possible subsets $\mathcal I_k$ of $\{1,\cdots, n\}$ with cardinality $k$.
\end{lem}
\begin{proof}
  The result is straightforward from the definition of $(p_j^{\downarrow})_{j=1}^n$.
\end{proof}

\begin{thm}\label{thm1}
Let $M= \{ |u_{j}\rangle \}_{j=1}^n$ and $N=\left\{ |v_{\ell}\rangle \right\}_{\ell=1}^n$ be the measurements of pre- and post-testing respectively. If the outcome probability distribution from $M$ is given by $\p = (c_j)_{j=1}^n$, then any outcome $\q$ from $N$ is bounded as $\left( \p ,  \rr \right) \prec_{R} \left( \p ,  \q \right) \prec_{R} \left( \p ,  \bm{s} \right)$ with
\begin{align}
\bm{s} &= \left(S_{1}, S_{2}, \ldots, S_{n} \right) :=\left(s_{1}, s_{2}-s_{1}, \ldots, s_{n}-s_{n-1} \right)
\quad \text{and}\quad 
s_{k} = \max\limits_{\mathcal{I}_{k}} \max\limits_{\rho}\left\{\,\Tr N_{\mathcal{I}_k}\,\rho \,\big|\ \rho \in S(M,\p)\,\right\},\\
\rr & = \left(R_{1}, R_{2}, \ldots, R_{n} \right) :=\left(r_{1}, r_{2}-r_{1}, \ldots, r_{n}-r_{n-1} \right)
\quad \text{and} \quad 
r_{k} = \min\limits_{\rho} \left\{\,\gamma \, \big|\ \gamma \geqslant \Tr N_{\mathcal{I}_k}\,\rho, \forall \mathcal{I}_{k},\ \rho \in S(M,\p)\,\right\},
\end{align}
where $\mathcal{I}_k$ denotes a subset of $\{1,\cdots,n\}$ with cardinality $k$ and $N_{\mathcal{I}_k}:= \sum_{\ell \in \mathcal{I}_{k}} | v_\ell \rangle\langle v_\ell |$ is a partial sum of $N$ with index $\mathcal I_k$.
\end{thm}
\begin{proof}
Suppose $\q :=(q_k)_{k=1}^n = (\langle v_{k}|\,\rho\,|v_{k}\rangle)_{k=1}^n$ is an outcome probability generated from the measurement $N$ on a quantum state $\rho \in S(M,\p)$. Denote the non-increasing order of $\q$ as $(q_j^{\downarrow})_{j=1}^n$.  By Lemma~\ref{majorization lemma} and the definitions of $q_s$ and $N_{\mathcal I_k}$, we have
\begin{align}\label{tmp4}
  \sum_{j=1}^k q_j^{\downarrow} = \max_{\mathcal I_k} \sum_{s\in \mathcal I_k} q_s = \max_{\mathcal I_k} \sum_{s\in \mathcal I_k} \Tr |v_s\rangle\langle v_s| \, \rho = \max_{\mathcal I_k} \Tr N_{\mathcal I_k} \, \rho.
\end{align}
Relaxing $\rho$ to all possible $\rho\in S(M,\p)$, we have
\begin{align}\label{up bound}
   \max_{\mathcal I_k} \Tr N_{\mathcal I_k} \, \rho \leqslant \max\limits_{\rho\in S(M,\p)} \max\limits_{\mathcal{I}_{k}} \Tr N_{\mathcal{I}_k}\,\rho =: s_{k}.
\end{align}
Note that $s_{k} = \sum_{j=1}^k S_{k} \leqslant \sum_{j=1}^k S_{k}^{\downarrow}$. Thus we have $\sum_{j=1}^k q_j^{\downarrow} \leqslant \sum_{j=1}^k S_{k}^{\downarrow}$. By the definition of majorization, we have $\q \prec \bm{s}$.

Similar to Eq.~\eqref{up bound}, we have
\begin{align}\label{low bound}
   \max_{\mathcal I_k} \Tr N_{\mathcal I_k} \, \rho \geqslant \min\limits_{\rho\in S(M,\p)} \max\limits_{\mathcal{I}_{k}} \Tr N_{\mathcal{I}_k}\,\rho =: r_{k}
\end{align}
Now we are going to show that $R_{k}$ is in the nonincreasing order, i.e., $R_{k} \geqslant R_{k+1}$, or equivalently $2r_{k} \geqslant r_{k-1} + r_{k+1}$.
Suppose the minimum of $r_{k}$ is taken at $\widetilde \rho$. Denote $\widetilde \q = (\langle v_{j}|\,\widetilde \rho\,|v_{j}\rangle)_{j=1}^n$ with the nonincreasing ordering $\big(\,\widetilde q_j^{\downarrow}\,\big)_{j=1}^n$. Then we have
\begin{align}\label{tmp1}
  2r_{k} = 2 \max\limits_{\mathcal{I}_{k}} \Tr N_{\mathcal{I}_k}\, \widetilde \rho = 2 \sum_{j=1}^k \widetilde q_j^{\downarrow} = \left(\sum_{j=1}^{k-1} \widetilde q_j^{\downarrow} + \widetilde q_{k}^{\downarrow}\right) + \left(\sum_{j=1}^{k+1} \widetilde q_j^{\downarrow} - \widetilde q_{k+1}^{\downarrow}\right) \geqslant \sum_{j=1}^{k-1} \widetilde q_j^{\downarrow} + \sum_{j=1}^{k+1} \widetilde q_j^{\downarrow},
\end{align}
where the second equality follows from Lemma~\ref{majorization lemma}.  Using Lemma~\ref{majorization lemma} again, we have
\begin{align}
  \sum_{j=1}^{k-1} \widetilde q_j^{\downarrow} = \max\limits_{\mathcal{I}_{k-1}} \Tr N_{\mathcal{I}_{k-1}}\,\widetilde \rho \geqslant \min\limits_{\rho\in S(M,\p)} \max\limits_{\mathcal{I}_{k-1}} \Tr N_{\mathcal{I}_{k-1}}\,\rho = r_{k-1}, \label{tmp2}\\
  \sum_{j=1}^{k+1} \widetilde q_j^{\downarrow} = \max\limits_{\mathcal{I}_{k+1}} \Tr N_{\mathcal{I}_{k+1}}\,\widetilde \rho \geqslant \min\limits_{\rho\in S(M,\p)} \max\limits_{\mathcal{I}_{k+1}} \Tr N_{\mathcal{I}_{k+1}}\,\rho = r_{k+1}.\label{tmp3}
\end{align}
Combining Eqs.~(\ref{tmp1}-\ref{tmp3}), we have $2r_{k} \geqslant r_{k-1} + r_{k+1}$, that is, $R_{k}$ is in the nonincreasing order. Together with Eqs.~\eqref{tmp4} and~\eqref{low bound}, we obtain $\sum_{j=1}^k q_j^{\downarrow} \geqslant r_{k} = \sum_{j=1}^k R_{k} =  \sum_{j=1}^k R_{k}^{\downarrow}$, which implies $ \q \succ \rr$. Finally, we note that 
\begin{align}
   r_{k} = \min\limits_{\rho\in S(M,\p)} \max\limits_{\mathcal{I}_{k}} \Tr N_{\mathcal{I}_k}\,\rho = \min\limits_{\rho\in S(M,\p)} \min\left\{\gamma\,|\, \gamma \geqslant \Tr N_{\mathcal{I}_k}\,\rho, \forall \mathcal I_k\right\} = \min\limits_{\rho} \left\{\,\gamma \, \big|\ \gamma \geqslant \Tr N_{\mathcal{I}_k}\,\rho, \forall \mathcal{I}_{k},\ \rho \in S(M,\p)\,\right\},
\end{align}
where the rightmost is a semidefinite program.
\end{proof}

\begin{remark}
  Note that we have shown $\rr$ is in nonincreasing order. But this is not necessarily the case for $\bm{s}$ since there exist cases such that $2 s_{k} < s_{k-1} + s_{k+1}$. That is why we need an additional flatness procedure for the upper bound in general. However, for the qubit and qutrit case, we can show that the vector $\bm{s}$ is always in nonincreasing order.
\end{remark}

\begin{proposition}\label{p2}
  For $n = 2, 3$, the vector $\bm{s}$ is in nonincreasing order.
\end{proposition}
\begin{proof}
  Recall that $s_{k} = \max\limits_{\rho\in S(M,\p)} \max\limits_{\mathcal{I}_{k}} \Tr N_{\mathcal{I}_k}\,\rho$. We first show that $2s_{1} \geqslant s_{2}$.
Suppose the optimal state of $s_2 = \max\limits_{\rho\in S(M,\p)} \max\limits_{\mathcal{I}_{2}} \Tr N_{\mathcal{I}_2}\,\rho$ is taken at $\widetilde \rho$ and the optimal index set $\mathcal I_2 = \{1,2\}$. We have
\begin{align}
 s_{2} = \Tr (\ket{v_1}\bra{v_1}+\ket{v_2}\bra{v_2})\,
  \widetilde \rho \leqslant 2 \max\limits_{\mathcal{I}_1} \Tr N_{\mathcal{I}_1}\, \widetilde \rho \leqslant 2 \max\limits_{\rho\in S(M,\p)} \max\limits_{\mathcal{I}_1} \Tr N_{\mathcal{I}_1}\, \rho = 2s_{1}.
\end{align}
Furthermore, when $n=3$, we have $s_{3} = 1$. Suppose the optimal state of $s_{1}$ is taken at $\widetilde \rho$ and optimal index is taken at $\mathcal I_1 = \{1\}$. Then we have
\begin{align}
  s_{3} + s_{1} = 1 + \Tr \ket{v_1}\bra{v_1}\, \widetilde \rho & = \Tr (\ket{v_1}\bra{v_1} + \ket{v_2}\bra{v_2}) \, \widetilde \rho + \Tr (\ket{v_1}\bra{v_1} + \ket{v_3}\bra{v_3}) \, \widetilde \rho \leqslant 2 \max\limits_{\rho\in S(M,\p)} \max\limits_{\mathcal{I}_{2}} \Tr N_{\mathcal{I}_2}\,\rho = 2s_{2},
\end{align}
which completes the proof.
\end{proof}

\begin{remark}
  We show that when $n = 4$, there exists an explicit example that $\bm{s}$ is not in nonincreasing order. Consider the measurements
  \begin{align}
    M = \begin{pmatrix}
      -0.4703  & -0.3508  & -0.6040 &  -0.5394\\
    0.8392  &  0.0745 &  -0.4820  & -0.2404\\
    0.0627  &  0.2180 &  0.5452  & -0.8070\\
    0.2657  & -0.9077 &  0.3249  & -0.0051
    \end{pmatrix},\quad 
    N = \begin{pmatrix}
      0.4960  & 0.1166 &  -0.8579  & -0.0654\\
   -0.3299  & -0.2437 &  -0.2898  &  0.8647\\
    0.6759  & -0.6566 &   0.2886  &  0.1696\\
    0.4339  &  0.7042  &  0.3109  &  0.4682
    \end{pmatrix}
  \end{align}
  and the probability vector $\p=(
    1/4, \ 1/4, \ 1/4, \ 1/4)$. By running our SDP algorithm, we obatin the result $\bm{s} = (0.9047,\  0.0424,\  0.0529,\  0.0001)$ where the second element is strictly smaller than the third one.
\end{remark}

\section{Proof of Theorem 2}

In order to prove the tightness of bounds $\rr$ and $\ttt$ we first recall all the related concepts
\begin{definition}[Poset]
A partial order is a binary relation ``$\prec$'' over a set $\mathcal{L}$ satisfying reflexivity, antisymmetry, and transitivity. That is, for all $x$, $y$, and $z$ in $\mathcal{L}$, we have 
\begin{enumerate}[label=(\roman*)]
\item Reflexivity: $x \prec x$,
\item Antisymmetry: If $x \prec y$ and $y \prec x$, then $x=y$,
\item Transitivity: If $x \prec y$ and $y \prec z$, then $x \prec z$.
\end{enumerate}
\end{definition}
\noindent Note that without the antisymmetry, ``$\prec$'' is just a preorder. Let us now define the set of all $n$-dimensional probability vectors as
\begin{align}
\mathcal{P}^{n} = \left\{\p=\left(p_{1}, \ldots, p_{n}\right)~|~
p_{j}\in[0,1], \sum\limits_{j=1}^{n}p_{j}=1, p_{j}\geqslant p_{j+1}
\right\},
\end{align}
with components in non-increasing order. Accordingly, majorization is a partial order over $\mathcal{P}^{n}$, i.e. $\langle\mathcal{P}^{n}, \prec\rangle$ is a poset.

\begin{definition}[Lattice]
A poset $\langle \mathcal{L}, \prec \rangle$ is called a join-semilattice, if for any two elements $x$ and $y$ of $\mathcal{L}$, it has a unique least upper bound (lub,supremum) $x \lor y$ satisfying
\begin{enumerate}[label=(\roman*)]
\item $x \lor y \in \mathcal{L}$,
\item $x \prec x \lor y$ and $y \prec x \lor y$.
\end{enumerate}
On the other hand, $\langle \mathcal{L}, \prec \rangle$ is called a meet-semilattice, if for any two elements $x$ and $y$ of $\mathcal{L}$, it has a unique greatest lower bound (glb,infimum) $x \land y$ satisfying
\begin{enumerate}[label=(\roman*)]
\item $x \land y \in \mathcal{L}$,
\item $x \land y \prec x$ and $x \land y \prec y$.
\end{enumerate}
$\langle \mathcal{L}, \prec \rangle$ is called a lattice if it is both a join-semilattice and a meet-semilattice, and denote it as a quadruple $\langle \mathcal{L}, \prec, \land, \lor \rangle$.
\end{definition}

\begin{definition}[Bounded Lattice]
A lattice $\langle \mathcal{L}, \prec, \land, \lor \rangle$ is called bounded, if it has a top, denoted by $\top$ and a bottom, denoted by $\bot$ which satisfy
\begin{enumerate}[label=(\roman*)]
\item $x \prec \top$ \quad for all $x \in \mathcal{L}$,
\item $\bot \prec x$ \quad for all $x \in \mathcal{L}$.
\end{enumerate}
\end{definition}

\noindent By definition, $\mathcal{P}^{n}$ is bounded, since for any probability vector $\p$ belongs to the set $\mathcal{P}^{n}$ we have $\mathbf{u} \prec \p \prec \mathbf{l}$ where $\mathbf{u} := \left(1/n,\ldots,1/n\right)$ and $\mathbf{l} := \left(1,0,\ldots,0\right)$.

\begin{definition}[Complete Lattice]
A lattice $\langle \mathcal{L}, \prec, \land, \lor \rangle$ is called complete, if for any subset $\mathcal{S}\subset\mathcal{L}$, it has a greatest element, denoted by $\top$ and a least element, denoted by $\bot$ which satisfy 
\begin{enumerate}[label=(\roman*)]
\item $x \prec \top$, \, for all $x\in\mathcal{S}$ and 
$x \prec y$ \, for all $x\in\mathcal{S}$ $\Rightarrow \top \prec y$,
\item $\bot \prec x$, \, for all $x\in\mathcal{S}$ and 
$y \prec x$ \, for all $x\in\mathcal{S}$ $\Rightarrow y \prec \bot$.
\end{enumerate}
\end{definition}

\noindent By embedding the majorization ``$\prec$'', the quadruple $\langle \mathcal{P}^{n}, \prec, \land, \lor \rangle$ forms a complete lattice. We remark that the result of completeness follows directly from the work presented in \cite{Rapat1991SM}, and the algorithm in finding $\p \land \q$ and $\p \lor \q$ (also known as flatness process) was first introduced in \cite{Cicalese2002SM}. Only recently, the completeness of majorization lattice has been applied to the study of optimal common resource \cite{Bosyk2019SM} .

\begin{figure}
\centering
\begin{tikzpicture}
[minimum size=25pt,inner sep=0pt,
layer1/.style={circle,draw=cyan!80,fill=cyan!40,thick},
layer2/.style={circle,draw=green!80,fill=green!40,thick},
layer3/.style={circle,draw=blue!80,fill=blue!40,thick},
layer4/.style={circle,draw=red!80,fill=red!40,thick},
layer5/.style={circle,draw=magenta!80,fill=magenta!40,thick}] 
\node (11) at (-4,0) [layer1] {$\bot$};
\node (21) at (-2,1) [layer2] {$a \land b$};
\node (22) at (-2,-1) [layer2] {$d \land e$};
\node (31) at ( 0,2) [layer3] {a};
\node (32) at ( 0,1) [layer3] {b};
\node (33) at ( 0,0) [layer3] {c};
\node (34) at ( 0,-1) [layer3] {d};
\node (35) at ( 0,-2) [layer3] {e};
\node (41) at ( 2,2) [layer4] {$a \lor b$};
\node (42) at ( 2,0) [layer4] {$b \lor c$};
\node (43) at ( 2,-2) [layer4] {$d \lor e$};
\node (51) at ( 4,0) [layer5] {$\top$};
\draw [->] (21) to (11);
\draw [->] (22) to (11);
\draw [->] (31) to (21);
\draw [->] (33) to (11);
\draw [->] (32) to (21);
\draw [->] (34) to (22);
\draw [->] (35) to (22);
\draw [->] (41) to (31);
\draw [->] (41) to (32);
\draw [->] (42) to (32);
\draw [->] (42) to (33);
\draw [->] (43) to (34);
\draw [->] (43) to (35);
\draw [->] (51) to (41);
\draw [->] (51) to (42);
\draw [->] (51) to (43);
\end{tikzpicture}
\caption{(color online) A schematic diagram of lattice structure $\langle \mathcal{L}, \prec, \land, \lor \rangle$, consisting of a partially order set (poset) $\langle \mathcal{L}, \prec\rangle$ and two binary operations $\land$,  and $\lor$. For any two elements of $\mathcal{L}$, there exists a unique greatest lower bound (glb, $\land$) and a unique least upper bound (lub, $\lor$) under $\prec$. Here $\mathcal{L} := \left\{\bot, a\land b, d\land e, a, b, c, d, e, a\lor b, b\lor c, d\lor e, \top\right\}$, and the notation ``$\leftarrow$'' stands for ``$\prec$''. Clearly, the quadruple $\langle \mathcal{L}, \prec, \land, \lor \rangle$ is a bounded lattice, that is for any $x\in\mathcal{L}$ we have $\bot\prec x\prec\top$. However, this lattice is not complete since by defining the subset $\mathcal{S}$ as $\left\{a\land b,a,b,c,a\lor b, b\lor c\right\}$, its greatest lower bound and the least upper bound do not exist.}
\label{lattice}
\end{figure}
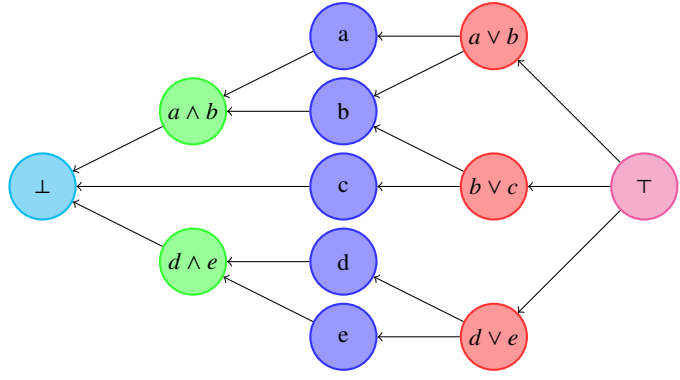

The main progress over the main text is a fine-grained approach in studying the feasible set of probability vector $\q$ conditioned on $\p$; that is
\begin{align}
\mathcal{Q}:=\left\{\q :=\left(\Tr(|v_{k}\rangle\langle v_{k}|\,\rho)\right)_{k} ~|~ \rho \in S(M,\p) \right\},
\end{align}
which forms a convex subset of $\mathcal{P}^{n}$ and our aim is to find its infimum and supremum. By appealing to the definition of lattice, it is straightforward to show that the unique infimum $\BigWedge\mathcal{Q}$ and supremum $\BigVee\mathcal{Q}$ of the set $\mathcal{Q}$ always exist whenever $\mathcal{Q}$ is a finite subset of $\mathcal{P}^{n}$. However, the feasible set $\mathcal{Q}$ will not be a finite subset of $\mathcal{P}^{n}$ in general. A crucial property to solve this problem is the completeness of majorization lattice, and several useful definitions are needed. 
Let $\mathbb{R}^{n}_{+}$ and $\mathbb{R}^{n}_{+.\geqslant}$ be defined as
\begin{align}
\mathbb{R}^{n}_{+} =  
\left\{\p=\left(p_{1}, \ldots, p_{n}\right) \in \mathbb{R}^{n}~|~
p_{j}\geqslant 0\right\},
 \quad\text{and}\quad
 \mathbb{R}^{n}_{+,\geqslant} =  \left\{\p=\left(p_{1}, \ldots, p_{n}\right) \in \mathbb{R}^{n}~|~
p_{j}\geqslant p_{j+1} \geqslant 0
\right\}.
\end{align}
A vector $\x \in \mathbb{R}^n$ is weakly majorized by $\y\in \mathbb{R}^n$, denoted by $x\prec_{w}y$, if $\sum_{j=1}^{k} x^{\downarrow}_{j} \leqslant \sum_{j=1}^{k} y^{\downarrow}_{j}$ for all $1\leqslant k \leqslant n$. Here the down-arrow notation denotes that the components of the corresponding vector are ordered in an nonincreasing order. Based on these definitions, a stronger version of completeness has been proved in \cite{Rapat1991SM}
\begin{lem}[Infimum]
Let $\mathcal{S} \subset \mathbb{R}^{n}_{+}$ be a nonempty set. Then there exists a unique vector $\BigWedge\mathcal{S} \in  \mathbb{R}^{n}_{+,\geqslant}$ such that
\begin{enumerate}[label=(\roman*)]
\item $x\in\mathcal{S} \Rightarrow \BigWedge\mathcal{S}\prec_{w}x$,
\item $y \prec_{w} x \,\,\text{for all} \,\, x\in\mathcal{S} \Rightarrow y \prec_{w} \BigWedge\mathcal{S}$.
\end{enumerate}
\end{lem}
\begin{lem}[Supremum]
Let the set of upper bounds of $\mathcal{S} \subset \mathbb{R}^{n}_{+}$ be nonempty. Then there exists a unique vector $\BigVee\mathcal{S}$ such that
\begin{enumerate}[label=(\roman*)]
\item $x\in\mathcal{S} \Rightarrow x\prec_{w}\BigVee\mathcal{S}$,
\item $x \prec_{w} y \,\,\text{for all} \,\, x\in\mathcal{S} \Rightarrow \BigVee\mathcal{S} \prec_{w} y$.
\end{enumerate}
\end{lem}
Therefore, the completeness of majorization lattice follows from lemma 2 and lemma 3 immediately.
\begin{proposition}
  Majorization lattice is complete: for any subset $\mathcal{S}$ of $\mathcal{P}^{n}$, it has both a infimum and supremum in $\mathcal{P}^{n}$.
\end{proposition}
\begin{proof}
Taking any subset $\mathcal{S}$ of $\mathcal{P}^{n}$ we get that $\mathcal{S} \subset \mathcal{P}^{n} \subset \mathbb{R}^{n}_{+}$. The existence of $(1,0,\ldots,0)$ implies that the set of upper bounds of $\mathcal{S}$ is nonempty. Therefore, $\mathcal{S}$ must have a  infimum $\BigWedge\mathcal{S}$ and supremum $\BigVee\mathcal{S}$. We now prove that both $\BigWedge\mathcal{S}$ and $\BigVee\mathcal{S}$ belongs to $\mathcal{P}^{n}$. Since $\mathcal{P}^{n}$ is bounded by $(1/n,1/n,\ldots,1/n)$ and $(1,0,\ldots,0)$, hence for any $x \in \mathcal{S}$ we have
\begin{align}
(1/n,1/n,\ldots,1/n) \prec x \prec (1,0,\ldots,0).
\end{align}
Therefore
\begin{align}
(1/n,1/n,\ldots,1/n) \prec \BigWedge\mathcal{S} \prec x \prec \BigVee\mathcal{S} \prec (1,0,\ldots,0). \quad \forall x\in\mathcal{S}
\end{align}
which is equivalent to say $\BigWedge\mathcal{S}$, $\BigVee\mathcal{S}\in\mathcal{P}^{n}$. We thus prove the completeness of majorization lattice.
\end{proof}
From the completeness of $\mathcal{P}^{n}$, we can always find the unique infimum $\BigWedge\mathcal{Q}$ and supremum $\BigVee\mathcal{Q}$ for the feasible set $\mathcal{Q}$. Now we are in position to prove that $\BigWedge\mathcal{Q}=\rr$ and $\BigVee\mathcal{Q}=\ttt$.

\begin{thm}
For pre-testing with outcome probability distribution $\p = (c_j)_{j=1}^n$, the outcome $\q$ of post-testing is bounded by $\rr$ and $\ttt$, and they are tight under majorization; That is
\begin{align}
  \BigWedge\mathcal{Q} = \rr \prec \q \prec 
  \ttt = \BigVee\mathcal{Q}.
\end{align}
\end{thm}
\begin{proof}
We start by proving $\BigWedge\mathcal{Q} = \rr$. From the definition of the infimum $\BigWedge\mathcal{Q}$ we have $\rr\prec \BigWedge\mathcal{Q}$. Let $\BigWedge\mathcal{Q} := \left(x_{1},\ldots,x_{n}\right)$, then it is immediate to observe that $\BigWedge\mathcal{Q} \prec \q$ for any $\q \in \mathcal{Q}$; that is
\begin{align}
\sum\limits_{s=1}^{k} x_{s} \leqslant 
\max\limits_{\mathcal{I}_{k}} \Tr N_{\mathcal{I}_{k}}\,\rho,
\quad \forall k\in \left\{1,\ldots,n\right\} \quad\text{and}\quad \forall \rho\in S(M,\p)
\end{align}
which leads to
\begin{align}
\sum\limits_{s=1}^{k} x_{s} \leqslant 
\min\limits_{\rho\in S(M,\p)} \max\limits_{\mathcal{I}_{k}} \Tr N_{\mathcal{I}_{k}}\,\rho = r_{k},
\quad \forall k\in \left\{1,\ldots,n\right\}
\end{align}
and hence one has $\BigWedge\mathcal{Q} \prec \rr$. Since both $\BigWedge\mathcal{Q}$ and $\rr$ (see the proof of theorem \ref{thm1}) are arranged in nonincreasing order, $\BigWedge\mathcal{Q} = \rr$ holds.  

On the other hand, in order to prove that $\ttt$ is indeed the supremum $\BigVee\mathcal{Q}$, we need only to prove that $\ttt \prec \BigVee\mathcal{Q}$. Suppose now $\BigVee\mathcal{Q}$ has the form $\left(y_{1},\ldots,y_{n}\right)$, and assume, by contradiction, that there exists an index $l$ such that
\begin{align}
\sum\limits_{s=1}^{l} y_{s} < t_{l},
\end{align}
where $t_{k} := \sum_{s=1}^{k} T_{s}$ ($k=1,\ldots,n$). Without loss of generality we can assume $l$ is the smallest index for which above inequality holds, and then we have $\sum\limits_{s=1}^{l-1} y_{s} \geqslant t_{l-1}$, which implies $y_{l}<a$. 

If $l<i$, the summations of $\BigVee\mathcal{Q}$ with the upper index smaller than $l+1$ yields
\begin{align}
t_{l} = s_{l} =  \max\limits_{\mathcal{I}_{l}} \max\limits_{\rho}\left\{\,\Tr N_{\mathcal{I}_l}\,\rho \,\big|\ \rho \in S(M,\p)\,\right\} \leqslant \sum\limits_{s=1}^{l} y_{s} 
<t_{l} = s_{l}.
\quad \forall l\in \left\{1,\ldots,i-1\right\}
\end{align}
Hence, under the standing hypothesis on $l$, we must have $l\geqslant i$.

Consider now $l>j-1$, from the flatness process we have $t_{l} = s_{l}$. Therefore, for index $l$ it holds that
\begin{align}
t_{l} = s_{l} =  \max\limits_{\mathcal{I}_{l}} \max\limits_{\rho}\left\{\,\Tr N_{\mathcal{I}_l}\,\rho \,\big|\ \rho \in S(M,\p)\,\right\} \leqslant \sum\limits_{s=1}^{l} y_{s} 
<t_{l} = s_{l},
\quad \forall l\in \left\{j,\ldots,n\right\}
\end{align}
that is a contradiction. Thus we conclude that $l\in \left\{i,\ldots,j-1\right\}$.

Finally, for all $l = i, i+1, \ldots, j-1$, one has that
\begin{align}
t_{j} =s_{j} =  \max\limits_{\mathcal{I}_{j}} \max\limits_{\rho}\left\{\,\Tr N_{\mathcal{I}_j}\,\rho \,\big|\ \rho \in S(M,\p)\,\right\} \leqslant \sum\limits_{s=1}^{j} y_{s} = \sum\limits_{s=1}^{l} y_{s} + \sum\limits_{s=l+1}^{j} y_{s} <
t_{l} +  \sum\limits_{s=l+1}^{j} y_{s}.
\end{align}
For $l<s\leqslant j$, we also have $a>y_{l}\geqslant\ldots\geqslant y_{j}$. As an immediate consequence, we obtain $\sum_{s=l+1}^{j} y_{s}<(j-l)a$. Hence
\begin{align}
 \sum\limits_{s=1}^{j} y_{s} <
t_{l} + (j-l)a = t_{j}.
\end{align}
So that $t_{j}<t_{j}$, which is clearly a contradiction. Therefore, $\ttt \prec \BigVee\mathcal{Q}$. Consequently, $\ttt , \BigVee\mathcal{Q} \in \mathcal{P}^{n}$ implies that $\ttt = \BigVee\mathcal{Q}$, which completes the proof.
\end{proof}

\section{Example 1: Uncertainty relation between two measurements}

In this part, we derive the analytical result of our majorization bounds for two measurements on a qubit state.
Without loss of generality, we consider two sets of qubit measurements 
\begin{equation}
\begin{aligned}
  M  & = \left\{ \ket{u_1},\ket{u_2}\right\}\quad \text{with} \quad \ket{u_1} = \ket{0}, \ket{u_2} = \ket{1},\\
  N  & = \left\{\ket{v_1}, \ket{v_2}\right\} \quad \text{with} \quad \ket{v_1} = \cos \theta\ket{0} - \sin \theta \ket{1}, \ket{v_2} = \sin \theta \ket{0} + \cos \theta \ket{1}, \theta \in [0,\pi/2].
\end{aligned}
\label{qubit measurement definition}
\end{equation}

\begin{proposition}\label{qubit proposition}
Suppose the qubit measurements $M$, $N$ are defined above.
If the outcome probability distribution under the measurement $M$ is $\p = (\lambda, 1-\lambda)$ with $\lambda \in (0,1/2)$, then the outcome $\q$ under the measurement $N$ is bounded as $\rr \prec \q \prec \bm{s}$ with $\bm{s} = (s_{1}, 1-s_{1})$, $\rr = (r_{1}, 1-r_{1})$ and
\begin{align}
  s_{1}  = \begin{cases}
    \big(\sqrt{\lambda} \sin \theta + \sqrt{1-\lambda} \cos \theta\big)^2, & \theta \in \big[0,\frac{\pi}{4}\big),\\[25pt]
    \big(\sqrt{1-\lambda} \sin \theta + \sqrt{\lambda} \cos \theta\big)^2, & \theta \in \big[\frac{\pi}{4},\frac{\pi}{2}\big],
  \end{cases}\quad
  r_{1} = \begin{cases}
    (\sqrt{1-\lambda}\sin \theta - \sqrt{\lambda} \cos\theta)^2, & \cot(2\theta) \in \Big(-\infty, - \frac{2\sqrt{\lambda(1-\lambda)}}{1-2\lambda}\Big),
    \\[5pt]
    \hspace{2cm}\frac{1}{2}, & \cot(2\theta) \in \Big[- \frac{2\sqrt{\lambda(1-\lambda)}}{1-2\lambda}, \frac{2\sqrt{\lambda(1-\lambda)}}{1-2\lambda}\Big],\\[5pt]
    \big(\sqrt{\lambda}\sin \theta - \sqrt{1-\lambda} \cos\theta\big)^2, & \cot(2\theta) \in \Big(\frac{2\sqrt{\lambda(1-\lambda)}}{1-2\lambda},\infty\Big).
  \end{cases}
\end{align}
\end{proposition}

\begin{proof}
The proof can be given by straightforward calculations.
According to Theorem 1, we have 
\begin{align}
  s_{1} = \max \{k_1,k_2\}, \ \text{with}\ \ k_s =  \underset{\rho}{\text{maximize}} &\ \Tr \ket{v_s}\bra{v_s} \, \rho \\
   \text{subject to} &\ \ \rho \geqslant 0,\, \langle0|\rho|0\rangle = \lambda,\, \langle1|\rho|1\rangle = 1-\lambda. \label{condition rho 1}
\end{align}
Since $\ket{v_s}$ are real vectors, the optimal solution can be always taken at a real operator $\rho$. From the condition~\eqref{condition rho 1}, let $\rho = \lambda \ket{0}\bra{0} + (1-\lambda)\ket{1}\bra{1} + x (\ket{0}\bra{1}+\ket{1}\bra{0})$. We have
\begin{align}
  k_1 = \underset{x}{\text{maximize}} &\ (1-2\lambda ) \sin^2\theta+\lambda - x\sin 2\theta\\
  \text{subject to} &\ -\sqrt{\lambda(1-\lambda)} \leqslant x \leqslant \sqrt{\lambda(1-\lambda)}.
\end{align}
Thus $k_1 = \big(\sqrt{1-\lambda} \sin \theta + \sqrt{\lambda} \cos \theta\big)^2$. 
Similarly, we have
$k_2 =  \big(\sqrt{\lambda} \sin \theta + \sqrt{1-\lambda} \cos \theta\big)^2$. Then we have
  \begin{align}
  s_{1} = \begin{cases}
    \big(\sqrt{\lambda} \sin \theta + \sqrt{1-\lambda} \cos \theta\big)^2, & \theta \in \big[0,\frac{\pi}{4}\big),\\[5pt]
    \big(\sqrt{1-\lambda} \sin \theta + \sqrt{\lambda} \cos \theta\big)^2, & \theta \in \big[\frac{\pi}{4},\frac{\pi}{2}\big].
  \end{cases}
\end{align}
Again based on Theorem 1, we have
\begin{align}
  r_{1} = \underset{\rho}{\text{minimize}} &\ \gamma\\
  \text{subject to} &\ \gamma \geqslant \Tr \ket{v_1}\bra{v_1}\bigcdot \rho,\ \gamma \geq \Tr \ket{v_2}\bra{v_2}\bigcdot \rho,\\
  &\ \rho \geqslant 0,\, \langle0|\rho|0\rangle = \lambda,\, \langle1|\rho|1\rangle = 1-\lambda.
\end{align}
Let $\rho = \lambda \ket{0}\bra{0} + (1-\lambda)\ket{1}\bra{1} + x (\ket{0}\bra{1}+\ket{1}\bra{0})$. We have
\begin{align}
  r_{1}= \underset{x}{\text{minimize}}  &\ \gamma_1 (x) \vee \gamma_2 (x)\\
  \text{subject to} &\ -\sqrt{\lambda(1-\lambda)} \leqslant x \leqslant \sqrt{\lambda(1-\lambda)},
\end{align}
with $\gamma_1(x)=(1-2\lambda ) \sin ^2 \theta +\lambda- x\sin 2\theta$ and $\gamma_2(x) = (1-2\lambda) \cos^2 \theta + \lambda + x \sin 2\theta$ and  $\gamma_1 (x) \vee \gamma_2 (x)$ denotes the  maximum function  between $\gamma_1$ and $\gamma_2$. Note that $\gamma_2(x) - \gamma_1(x) = 2\sin 2\theta \big[x - (\lambda-\frac12) \cot 2\theta\big]$.
If $(\lambda-\frac12) \cot(2\theta) \leqslant -\sqrt{\lambda(1-\lambda)}$, we have $\gamma_2\geqslant \gamma_1$ for all feasible $x$, then 
\begin{align}
r_{1} = \gamma_2\left(-\sqrt{\lambda(1-\lambda)}\right) = \left(\sqrt{\lambda}\sin \theta - \sqrt{1-\lambda} \cos\theta\right)^2.
\end{align}
If $(\lambda-\frac12) \cot(2\theta) \geqslant \sqrt{\lambda(1-\lambda)}$, we have $\gamma_2 \leqslant \gamma_1$ for all feasible $x$, then 
\begin{align}
  r_{1} = \gamma_1\left(\sqrt{\lambda(1-\lambda)}\right) = \left(\sqrt{1-\lambda}\sin \theta - \sqrt{\lambda} \cos\theta\right)^2.
\end{align}
If $ - \sqrt{\lambda(1-\lambda)}\leqslant (\lambda-\frac12) \cot(2\theta) \leqslant \sqrt{\lambda(1-\lambda)}$, the value $r_{1}$ is taken at the intersection point of $\gamma_1(x) = \gamma_2(x)$ and we will always have $r_{1} = \frac{1}{2}$.

\end{proof}

\begin{remark}
  Note that for the probability vectors $(\lambda,1-\lambda)$ and $(1-\lambda,\lambda)$, we will have exactly the same result of $s_{1}$ and $r_{1}$. 
\end{remark}

\begin{figure}[t]
\centering
\begin{tikzpicture}
  \node[inner sep=0pt] at (-12,2.2) {\includegraphics[width=3.6cm]{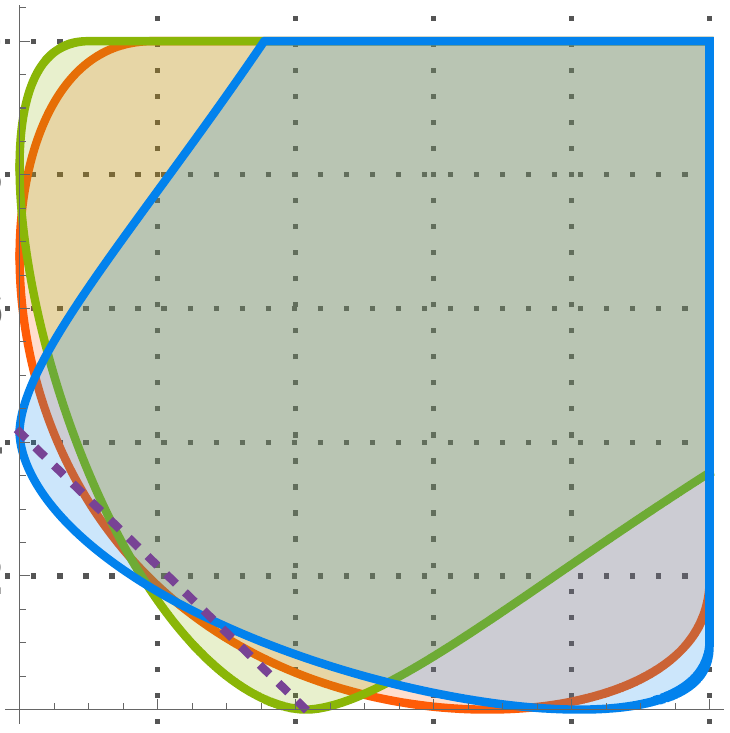}};
  \node[inner sep=0pt] at (-7,2.2) {\includegraphics[width=3.6cm]{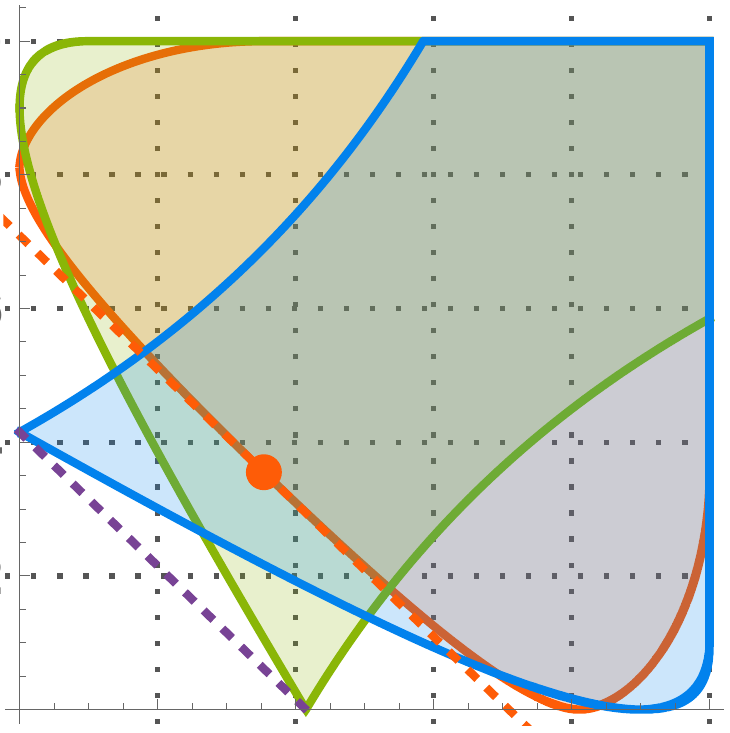}};
  \node[inner sep=0pt] at (-2,2.2) {\includegraphics[width=3.6cm]{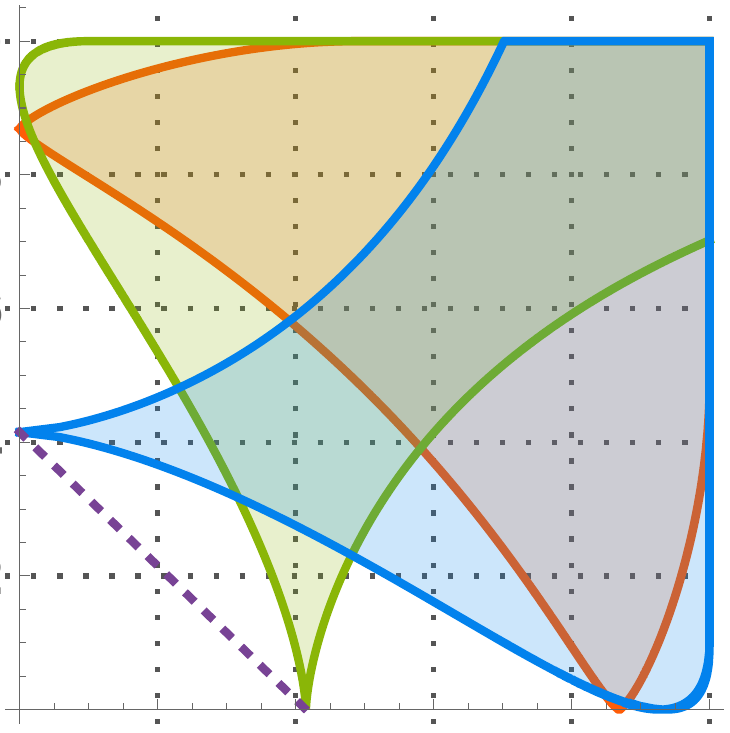}};
  \node at (-12,-0.2) {\small (a) $1/\alpha + 1/\beta = 1$};
  \node at (-7,-0.2) {\small (b) $1/\alpha + 1/\beta = 2$};
  \node at (-2,-0.2) {\small (c) $1/\alpha + 1/\beta = 3$};

  \draw[colorone,very thick] (0.6,3.3) --  (1.4,3.3);
  \node at (2.6,3.3) {\small $(\alpha,\beta)=\left(\frac{2}{c},\frac{2}{c}\right)$};

  \draw[colortwo,very thick] (0.6,2.5) --  (1.4,2.5);
  \node at (2.6,2.5) {\small $(\alpha,\beta)=\left(\infty,\frac{1}{c}\right)$};

  \draw[colorthree,very thick] (0.6,1.7) --  (1.4,1.7);
  \node at (2.6,1.7) {\small $(\alpha,\beta)=\left(\frac{1}{c},\infty\right)$};

  \draw[colorfour,line width=0.5mm,dotted] (0.6,0.9) --  (1.4,0.9);
  \node at (2.6,0.9) {\small MU bound};

  \node at (-13,0.25) {\footnotesize $0.2$};
  \node at (-12.35,0.25) {\footnotesize $0.4$};
  \node at (-11.65,0.25) {\footnotesize $0.6$};
  \node at (-10.95,0.25) {\footnotesize $0.8$};
  \node at (-10.27,0.25) {\footnotesize $1$};

  \node at (-14,1.15) {\footnotesize $0.2$};
  \node at (-14,1.8) {\footnotesize $0.4$};
  \node at (-14,2.5) {\footnotesize $0.6$};
  \node at (-14,3.15) {\footnotesize $0.8$};
  \node at (-14,3.8) {\footnotesize $1$};

  \node at (-8,0.25) {\footnotesize $0.2$};
  \node at (-7.35,0.25) {\footnotesize $0.4$};
  \node at (-6.65,0.25) {\footnotesize $0.6$};
  \node at (-5.95,0.25) {\footnotesize $0.8$};
  \node at (-5.27,0.25) {\footnotesize $1$};

  \node at (-9,1.15) {\footnotesize $0.2$};
  \node at (-9,1.8) {\footnotesize $0.4$};
  \node at (-9,2.5) {\footnotesize $0.6$};
  \node at (-9,3.15) {\footnotesize $0.8$};
  \node at (-9,3.8) {\footnotesize $1$};

  \node at (-3,0.25) {\footnotesize $0.2$};
  \node at (-2.35,0.25) {\footnotesize $0.4$};
  \node at (-1.65,0.25) {\footnotesize $0.6$};
  \node at (-0.95,0.25) {\footnotesize $0.8$};
  \node at (-0.27,0.25) {\footnotesize $1$};

  \node at (-4,1.15) {\footnotesize $0.2$};
  \node at (-4,1.8) {\footnotesize $0.4$};
  \node at (-4,2.5) {\footnotesize $0.6$};
  \node at (-4,3.15) {\footnotesize $0.8$};
  \node at (-4,3.8) {\footnotesize $1$};

\end{tikzpicture}
\caption{(colour online) Uncertainty regions with different parameters $1/\alpha + 1/\beta = c$; The red regions are uncertainty regions $\mathcal{R} \left( \mathrm{H}_{\alpha} , \mathrm{H}_{\beta} \right)$ with orders $(\alpha,\beta)=\left(2/c,2/c\right)$, the green regions stand for uncertainty regions with orders $(\alpha,\beta)=\left(\infty,1/c\right)$, the blue regions represent uncertainty regions with orders $(\alpha,\beta)=\left(1/c,\infty\right)$ and the dotted lines describe MU bound $\mathrm{H}_{\alpha} (M) + \mathrm{H}_{\beta} (N) \geqslant \log \left( 4/3 \right)$. Here, the horizontal axis represents the numerical value of $\mathrm{H}_{\alpha}$ and vertical axis stands for the numerical value of $\mathrm{H}_{\beta}$. (a) Uncertainty regions with $1/\alpha + 1/\beta = 1$; (b) Uncertainty regions with $1/\alpha + 1/\beta = 2$; (c) Uncertainty regions with $1/\alpha + 1/\beta = 3$.}
\label{MU}
\end{figure}

To illustrate the connection of our uncertainty regions with previously known results, we consider the most well-known Maassen and Uffink's (MU) entropic uncertainty relation \cite{Maassen1988SM}, which states that the R\'enyi entropies of measurement outcomes of $M$ and $N$ for any quantum state $\rho$ satisfy
\begin{align}
\label{exmu}
\mathrm{H}_{\alpha} (M) + \mathrm{H}_{\beta} (N) \geqslant -2\log c_{1},\quad \text{with} \quad c_{1} := \max_{j,k} | \langle u_{j} | v_{k} \rangle |,
\end{align}
where $\mathrm{H}_{\alpha} (M)$ and $\mathrm{H}_{\beta} (N)$ are the R\'enyi entropies of the probability distributions $\p$ and $\q$ with order $\alpha$ and $\beta$ respectively, and the orders $\alpha$ and $\beta$ satisfy $1/\alpha + 1/\beta = 2$.  Consider the measurements $M = \left\{ |0\rangle , |1\rangle \right\}$ and $N = \left\{ \frac{1}{2} |0 \rangle - \frac{\sqrt{3}}{2} |1\rangle , \frac{\sqrt{3}}{2} |0 \rangle + \frac{1}{2} |1\rangle \right\}$. Then the maximum overlap $c_1$ between $M$ and $N$ is $\sqrt{3}/2$ and hence MU relation holds for $1/\alpha + 1/\beta = 2$ with the lower bound $\log \left( 4/3 \right)$.

In this qubit case we have $\mathcal{R} \left( f , g \right)=\widetilde{\mathcal{R}} \left( f , g \right)$, and from Proposition~\ref{qubit proposition} we can study a family of uncertainty regions  $\mathcal{R} \left( \mathrm{H}_{\alpha} , \mathrm{H}_{\beta} \right)$ with $1/\alpha + 1/\beta = c$ for any $c > 0$. We first consider a family of uncertainty regions $\mathcal{R} \left( \mathrm{H}_{\alpha} , \mathrm{H}_{\beta} \right)$ with $1/\alpha + 1/\beta = 1$ in Fig. (\ref{MU}a) where MU bound losses its efficacy while our framework provides a full description for uncertainty regions with different orders. In the case $1/\alpha + 1/\beta = 2$ or $3$, MU bound is optimal for all allowable orders $\alpha$ and $\beta$. However, it is clearly not optimal for a specific pair of $\left( \alpha , \beta \right)$ except for the two extreme cases $(\alpha,\beta) = (1/2,\infty)$ and $(\infty,1/2)$. In Fig. (\ref{MU}b), for instance,  the uncertainty region with $\left( \alpha , \beta \right) = \left( 2/c , 2/c \right)$ is depicted in red, where the optimal uncertainty relation lower bound is given by a tangent line with slope $-1$ on the lower-left boundary which outperforms the MU bound. 

Actually the MU bound for qubit measurements is only tight in two extreme cases in general. For the qubit measurements $M$ and $N$ defined in~\eqref{qubit measurement definition}, the Maassen-Uffink bound is given by $-2\log c_1$ with
$c_1 = \cos \theta$ if $\theta \in \big[0,\frac{\pi}{4}\big)$ and $c_1 = \sin \theta$ if $\theta \in \big[\frac{\pi}{4},\frac{\pi}{2}\big]$.
When comparing with the Maassen-Uffink bound, we only need to consider the majorization upper bound $\bm{s}$, which corresponds to the lower-half boundary in the region plots. For the lower-half boundary
$\big(H_\alpha(\p), H_\beta(\bm{s})\big)$,
consider the zero point on the horizontal axis, i.e., $H_\beta(\bm{s}) = 0$. We have $\lambda = \cos^2 \theta$, when $\theta \in (0,\pi/4)$ and $\lambda = \sin^2 \theta$, when $\theta \in (\pi/4,\pi/2)$. Only when $\a \to \infty$, $H_\a(\p) \to H_{\infty}(\p) = -\log \max\{\lambda, 1-\lambda\} = -2 \log c_1$. That is, the zero point on the horizontal axis reaches the MU bound if and only if $\a = \infty$. Similar argument holds for $\beta$. Thus the MU bound for qubit measurements is only tight in the two extreme cases $(\alpha,\beta)=(1/2,\infty)$ or $(\infty,1/2)$. These observations indicate that our framework is generally more informative than the MU bound.

\section{Example 2: Uncertainty relation between multiple measurements}

Since the proof of the MU bound mainly relies on a version of {\it Riesz theorem} \cite{Riesz1926SM}, it would be difficult to generalize their result to multiple measurements or general $\alpha$, $\beta$ beyond the relation $1/\alpha + 1/\beta = 2$, even for the qubit case. 
As an demonstration of our framework to multiple measurements with general uncertainty measures, we consider the triple measurements $M$, $N$ (used in example 1) and $T = \left\{ \frac{\sqrt{2}}{2} |0 \rangle + \frac{\sqrt{2}}{2} |1\rangle , \frac{\sqrt{2}}{2} |0 \rangle - \frac{\sqrt{2}}{2} |1\rangle \right\}$. Denote $\bm{w}$ as the outcome probability distribution from measurement $T$. Then its R\'enyi entropy of order $\gamma$ is given by $\mathrm{H}_{\gamma} (T) := \mathrm{H}_{\gamma} (\bm{w})$. 
From Proposition~\ref{qubit proposition}, for any given probability $\p = (\lambda, 1-\lambda)$ with $\lambda\in (0,1/2)$, the majorization upper bound for $N$ and $T$ are respectively given by $\ttt^N = (t_1^N,1-t_1^N)$ and $\ttt^T=(t_1^T,1-t_1^T)$ with
\begin{align}
  t_1^N = \frac{\left(\sqrt{3}\sqrt{1-\lambda}+\sqrt{\lambda}\,\right)^2}{4}, \quad t_1^T = \frac{\left(\sqrt{1-\lambda}+\sqrt{\lambda}\,\right)^2}{2}.
\end{align}
For given $\alpha$, $\beta$, $\gamma$, the optimization
\begin{align}
  b:=\underset{\lambda}{\text{minimize}} \  H_{\alpha}(\p) + H_{\beta}(\ttt^N) + H_{\gamma}(\ttt^T)\quad 
  \text{subject to}\  0 < \lambda < \frac12,
\end{align}
leads to a state-independent lower bound for the uncertainty relation $H_{\alpha}(M) + H_{\beta}(N) + H_{\gamma}(T) \geq b$. 
Implementing this optimization via Mathematica function ``NMinimize'', we obtain the numerical results listed in the following table.

\setlength{\tabcolsep}{1em}
\begin{table}[H]
   \label{tab}
   \normalsize 
   \centering 
   \begin{tabular}{c c c c c c} 
   \toprule[\heavyrulewidth]\toprule[\heavyrulewidth]
   \textbf{Parameters $(\alpha,\beta,\gamma)$} & $(1,1,1)$ & $(1,2,3)$ & $(1,1,\infty)$ & $(\infty,1,\infty)$ & $(\infty,\infty,\infty)$\\
   \textbf{min $\mathrm{H}_{\alpha} (M) + \mathrm{H}_{\beta} (N) + \mathrm{H}_{\gamma} (T)$} & 1.15898 & 0.957202 & 0.903285 & 0.50165 & 0.474238\\ 
   \bottomrule[\heavyrulewidth] 
   \end{tabular}
   \caption{Numerical results for uncertainty relation $\mathrm{H}_{\alpha} (M) + \mathrm{H}_{\beta} (N) + \mathrm{H}_{\gamma} (T)$ consisting of three measurements $M$, $N$ and $T$. The parameters $(\alpha,\beta,\gamma)$ in the list are only chosen for demonstration. They could be any possible combinations of $\alpha$, $\beta$ and $\gamma$ in general.} 
\end{table}


\section{Example 3: Comparison with direct-sum and tensor-product bounds}

In this part we compare our results with the majorization uncertainty relations formulated in \cite{Gour2013SM,Rudnicki2013SM,Rudnicki2018SM,Rudnicki2014SM}. Explicit examples are provided to demonstrate how the obtained results in this work outperform previous ones.

Given quantum measurements $M = \{\ket{u_j}\}_{j=1}^n$ and $N = \{\ket{v_k}\}_{k=1}^n$ with $n$ possible classical outcomes, denote their joint set as $R := (\ket{r_1},\ket{r_2},\cdots,\ket{r_{2n}}) = (\ket{u_1},\cdots,\ket{u_n},\ket{v_1},\cdots,\ket{v_n})$ which contains $2n$ elements. Define
$x_k := \max_{\mathcal I_k} \|\sum_{s\in \mathcal I_k} \ket{r_s}\bra{r_s}\,\|_{\infty}$,
where the maximization is taken over all possible index subset $\mathcal I_k$ of $\{1,2,\cdots,2n\}$ with cardinality $k$.
The direct-sum upper bound is defined as \cite{Rudnicki2013SM}
\begin{align}
  \mathbf{w}_{\oplus}:= (x_1,x_2-x_1,\cdots,x_{2n} - x_{2n-1}) = (x_1,x_2-x_1,\cdots,x_{n+1} - x_{n}, \underbrace{0,0, \ldots, 0}_{\text{$n-1$ $~$times}}),
\end{align}
where the equation holds since $x_k = 2$ for $k \geqslant n+1$ \cite{Gour2013SM,Rudnicki2013SM}. Similarly, the direct-product upper bound is defined as \cite{Gour2013SM}
\begin{align}
  \mathbf{w}_{\otimes}:=\frac14 \Big(x_2^2,x_3^2-x_2^2,\cdots,x_{n+1}^2 - x_{n}^2,\underbrace{0,0, \ldots, 0}_{\text{$n^2-n$ $~$times}}\Big).
\end{align}
Denote the trivial lower and upper bounds as 
\begin{align}
  \mathbf{u} := \left(\frac{1}{n},\frac{1}{n},\cdots,\frac{1}{n}\right)\quad \text{and} \quad \mathbf{l}:= (1,0,\cdots,0).
\end{align}
In Fig.~\ref{mc}, we demonstrate that 
\begin{align}
  \mathbf{u}\oplus \mathbf{u} \prec \p \oplus \rr \prec & \p \oplus \q \prec \p \oplus \ttt \prec \mathbf{w}_\oplus \prec \mathbf{l} \oplus \mathbf{l},\\
  \mathbf{u}\otimes \mathbf{u} \prec \p \otimes \rr \prec & \p \otimes \q \prec \p \otimes \ttt \prec \mathbf{w}_\otimes \prec \mathbf{l} \otimes \mathbf{l},
\end{align}
where the probability vectors $\rr$ and $\ttt$ are obtained from Theorem 1. It is clear that our majorization upper bounds $\p \oplus \ttt$ and $\p \otimes \ttt$ provide the tighter estimation for $\p \oplus \q$ and $\p \otimes \q$ respectively. The shaded area depicts the feasible set of Lorenz curves given from $\mathcal{Q(M,N,\p)}$, whose boundaries are established by our majorization lower (purple) and upper (blue) bounds.


\begin{figure}[H]
\centering
\begin{tikzpicture}

  \draw[line width = 0mm, fill = gray!80, opacity = 0.5] (-6.85,-2.28) -- (-5.87,-0.3) -- (-3.8,2.2) -- (-2.8,2.6) -- (-1.8,2.8) -- (-1,2.78) -- (-1.7,2.6) -- (-4.8,0) -- (-6.81,-2.28);

  \draw[line width = 0mm, fill = gray!80, opacity = 0.5] (1.1,-2.4) -- (2.45,1.6) -- (3.8,2.4) -- (4.6,2.78) -- (5.3,2.8) -- (7.2,2.8) -- (5.25,2.45) -- (1.1,-2.4);

  \node[inner sep=0pt] at (-4,0) {\includegraphics[width=7cm]{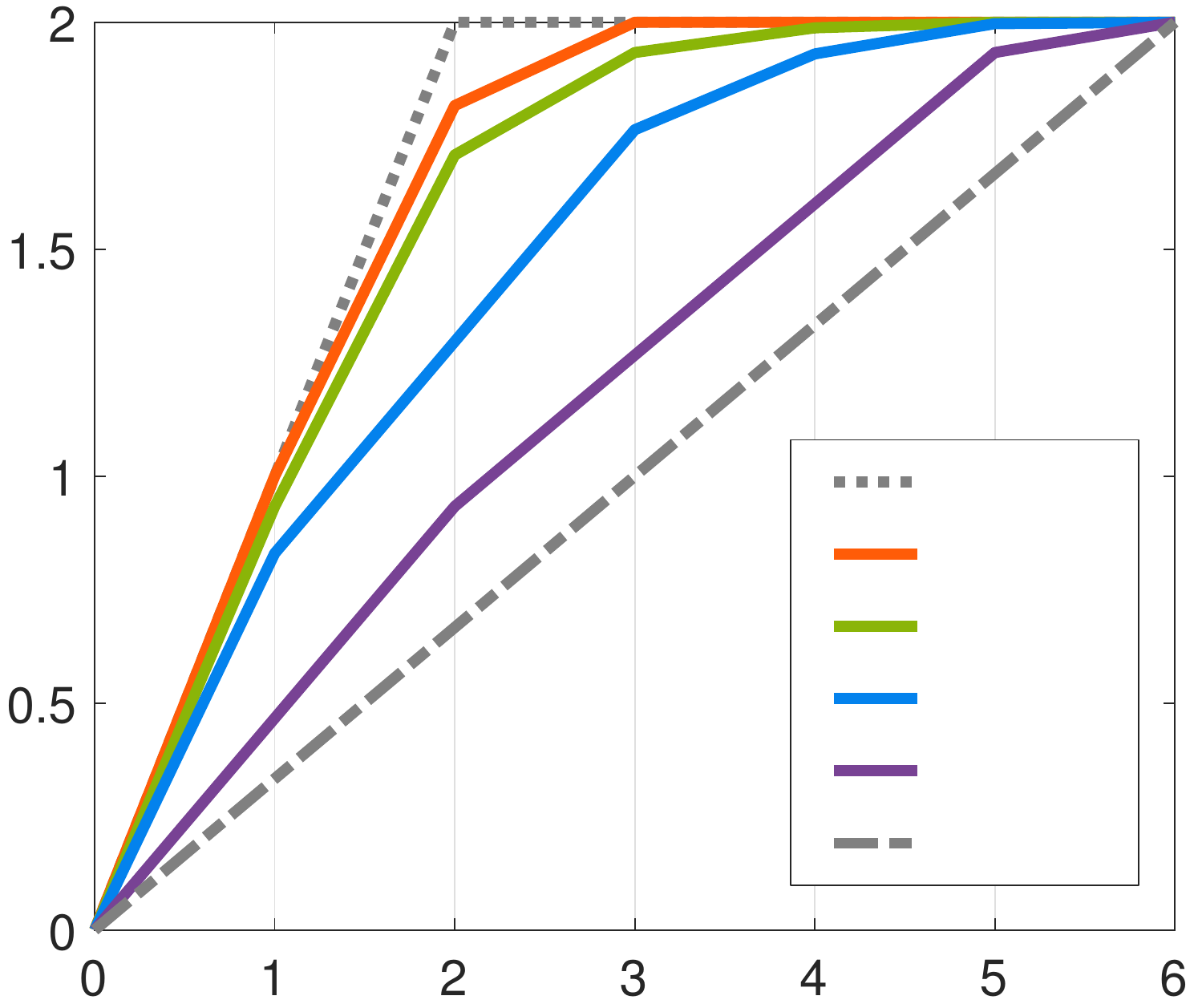}};
   \node[inner sep=0pt] at (4,0) {\includegraphics[width=7cm]{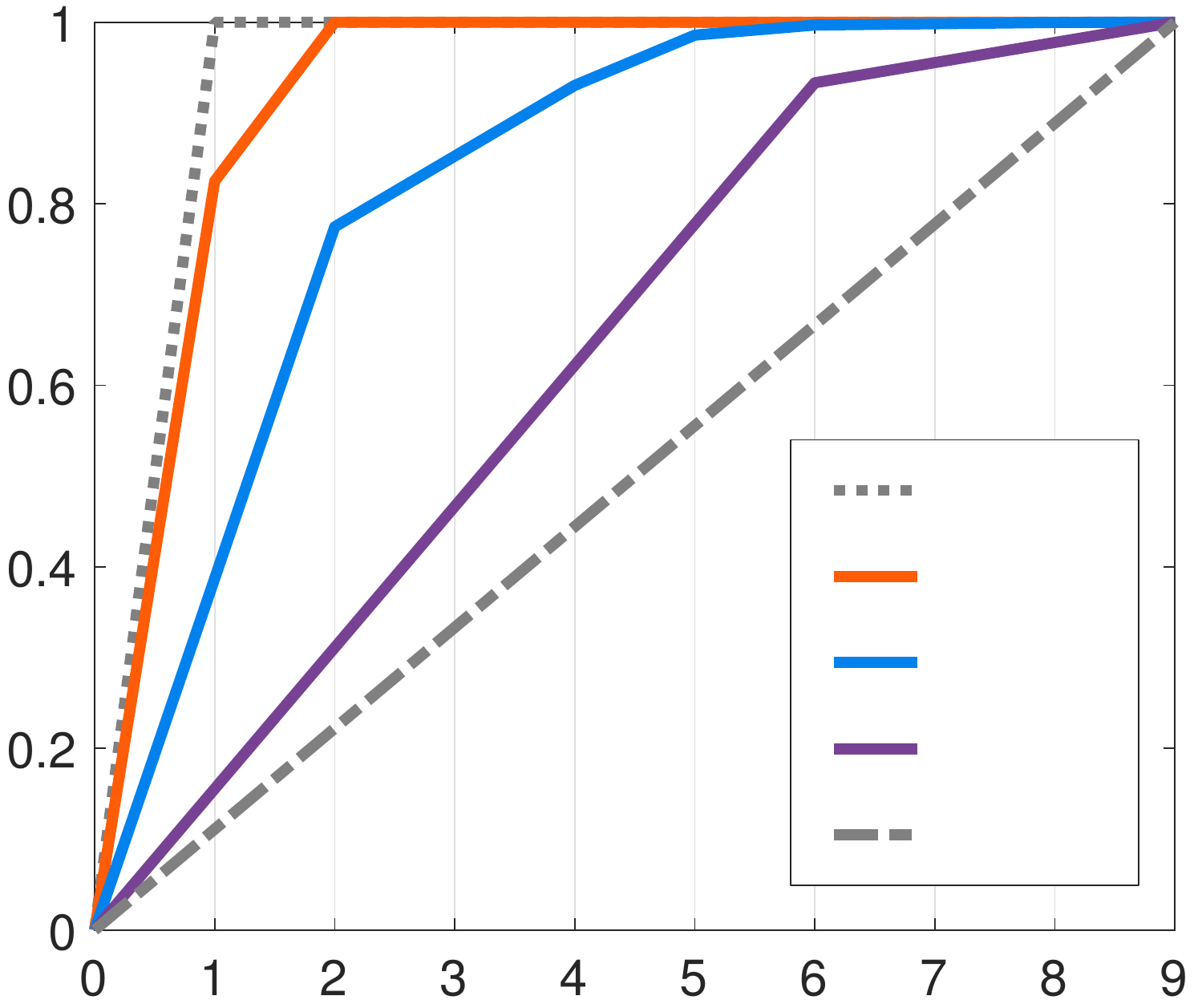}};


  \node[] at (-1.5,0.1) {$\mathbf{l} \oplus \mathbf{l}$};  
  \node[] at (-1.5,-0.3) {$\mathbf{w}_\oplus$};  
  \node[] at (-1.5,-0.7) {$\mathbf{w}_{\oplus} (\rho)$};  
  \node[] at (-1.45,-1.1) {$\p \oplus \ttt$};  
  \node[] at (-1.45,-1.55) {$\p \oplus \rr$};  
  \node[] at (-1.5,-1.95) {$\mathbf{u} \oplus \mathbf{u}$}; 

  \node[] at (-4,-3.3) {(a) Direct-sum comparison.};

  \node[] at (6.5,0.1) {$\mathbf{l} \otimes \mathbf{l}$};  
  \node[] at (6.5,-0.4) {$\mathbf{w}_\otimes$};  
  \node[] at (6.55,-0.9) {$\p \otimes \ttt$};  
  \node[] at (6.55,-1.45) {$\p \otimes \rr$};  
  \node[] at (6.5,-1.9) {$\mathbf{u} \otimes \mathbf{u}$};

  \node[] at (4,-3.3) {(b) Direct-product comparison.};

\end{tikzpicture}
\caption{(color online) Comparison of several known bounds via Lorenz curves. The example is taken as the probability $\p = (1,7,7)/15$, measurements $M = \{\,(1,0,0),(0,1,0),(0,0,1)\,\}$ and $N = \big\{\,\big(1/\sqrt{3},1/\sqrt{2},1/\sqrt{6}\big),\big(1/\sqrt{3},0,-\sqrt{2}/\sqrt{3}\big),\big(1/\sqrt{3},-1/\sqrt{2},1/\sqrt{6}\big)\,\big\}$.}
\label{mc}
\end{figure}

In the case that the spectrum of the quantum state $\rho$ is known, the direct-sum majorization bound can be improved further~\cite{Rudnicki2018SM}. Denote $\overrightarrow{\lambda(\rho)}$ and $\overrightarrow{\lambda(R\left(\mathcal I_k\right))}$ as the spectrum of $\rho$ and $R\left(\mathcal I_k\right):=\sum_{s\in \mathcal I_k} \ket{r_s}\bra{r_s}$ respectively (arrange in non-increasing order). Define $y_k := \max_{\mathcal I_k} \overrightarrow{\lambda(\rho)} \bigcdot \overrightarrow{\lambda(R\left(\mathcal I_k\right))}$, where the symbol $\bigcdot$ stands for inner product between vectors. Using the method presented in \cite{Rudnicki2014SM} one can show that, for the probability vectors $\p$ and $\q$ derived from measurements $M$ and $N$ on $\rho$, we get
\begin{align}
\p \oplus \q \prec \mathbf{w}_{\oplus} (\rho),\quad \text{with} \quad \mathbf{w}_{\oplus} (\rho):= (y_1,y_2-y_1,\cdots,y_{2n} - y_{2n-1}).
\end{align}
The same example in Fig.~(\ref{mc}a) demonstrates that the performance of our result can be still tighter than $\mathbf{w}_{\oplus} (\rho)$ where $\overrightarrow{\lambda(\rho)}$ is taken as the spectrum of $\rho = \frac{1}{15}\left(\begin{smallmatrix}   1 &  0 & 0\\
    0 & 7 & 7\\
    0 & 7 & 7
  \end{smallmatrix}\right)$. This indicates that the spectrum of a density operator is not necessary more useful than other classical outcomes of the same state, even though the former is usually regarded as more difficult to be obtained.


\section{Practical Resource Theories without Complete Tomography}

In quantum resource theories (QRTs) \cite{Chitambar2019SM}, one of the most central research topics is to study the conversion between different resource objects under certain constraints (or free operations) \cite{Plenio2007SM,Horodecki2009SM,Streltsov2017SM,Bennett1996PSM,Bennett1996CSM,Rains1999RSM,Rains1999BSM,Horodecki2000SM,Rains2001ASM,Waeldchen2016SM,Chitambar2018SM,Fang2017SM,Chitambar2016SM,Regula2018SM,Fang2018SM,Lami2019SM,Brandao2011SM,Gour2017SM,Zhao2018SM,Zhao2019SM,Wang2018SM,Liu2019SM}. In the recent study of resource theories, it is usually assumed that the density operator of our resource state is already known, which is not completely  practical. This is because quantum tomography is the main method used to assess the matrix representation of a quantum state, but the resources needed in achieving quantum tomography is exponential in the device size \cite{Sugiyama2013SM}. In particular, it is notoriously hard to determine the density matrix in high-dimensional Hilbert spaces and that is why we need to consider a more general framework in determining quantum state transformation with only partial knowledge of the state. 

We first consider QRT of entanglement. Let $|\psi\rangle^{AB}=\sum_{j=1}^{n}\sqrt{x_{j}}|j\rangle^{A}|j\rangle^{B}$ and $|\phi\rangle^{AB}=\sum_{j=1}^{n}\sqrt{y_{j}}|j\rangle^{A}|j\rangle^{B}$ be two  bipartite pure states, and their corresponding Schmidt vectors are defined as $\x := \left(x_{1}, \ldots, x_{n}\right)$ and $\y := \left(y_{1}, \ldots, y_{n}\right)$. By Nielsen's theorem \cite{Nielsen1999SM}, the necessary and sufficient condition for state transformation under local operations and classical communication (LOCC) is given as
\begin{align}
\psi \xrightarrow{\text{LOCC}} \phi \quad\Leftrightarrow\quad 
\x \prec \y  \quad(\text{Entanglement}).
\end{align}
This majorization criterion for entanglement transformation is fundamental and important since it plays an important roles in both deterministic and nondeterministic LOCC transformations, and gives rise to the concept of resource catalyst. Moreover, majorization criterion can also be found in QRT of coherence. For pure states $|\psi\rangle=\sum_{j=1}^{n}\sqrt{x_{j}}e^{i\psi_{j}}|j\rangle$ and $|\phi\rangle=\sum_{j=1}^{n}\sqrt{y_{j}}e^{i\phi_{j}}|j\rangle$, their probability amplitudes are denoted as $\x := \left(x_{1}, \ldots, x_{n}\right)$ and $\y := \left(y_{1}, \ldots, y_{n}\right)$ respectively. Similarly we have \cite{Winter2016SM,Zhu2017SM}
\begin{align}
\psi \xrightarrow[\text{IO}]{\text{SIO}} \phi \quad\Leftrightarrow\quad 
\x \prec \y  \quad(\text{Coherence}).
\end{align}
Here SIO stands for strictly incoherent operations and IO denotes incoherent operations. Inspired by these criteria, we wil focus on our study on QRTs for which
\begin{align}\label{QRT-M}
\psi \xrightarrow{\text{free}} \phi \quad\Leftrightarrow\quad 
\x\left(\psi\right) \prec \x\left(\phi\right) \quad(\text{Majorization-Based}),
\end{align}
where $\x\left(\psi\right)$ is a probability vector associated with $\psi$ (e.g. Schmidt vector). 

Let us now relate our complementary information relation with majorization-based QRTs \cite{Bosyk2019SM,Bosyk2017SM}, such as entanglement and coherence. 
If the pre-testing measurement $M$ can be implemented on our resource pure state $\psi$ via a free operation and outcomes a probability vector $\p$, then based on $\p$ it is possible to infer some useful results of resource conservation. More specifically, if we can derive the probability vector $\x\left(\psi\right)$ from post-testing with measurement $N$, then based on Theorem 2 we obtain the following result.

\begin{proposition}
If a pure state $\psi$ with free pre-testing indicates an outcome probability distribution $\p = (c_j)_j$, then a quantum state $\phi_{2}$ can be convert to $\psi$ by free operations if $\x\left(\phi_{2}\right) \prec \rr$  and $\psi$ can be transformed into $\phi_{1}$ whenever $\ttt \prec \x\left(\phi_{1}\right)$; that is
\begin{align}\label{QRT-M1}
\phi_{2} \xrightarrow{\text{free}} \psi \quad\Leftarrow\quad 
\x\left(\phi_{2}\right) \prec \rr,\quad \text{and} \quad \psi \xrightarrow{\text{free}} \phi_{1} \quad\Leftarrow\quad 
\ttt \prec \x\left(\phi_{1}\right).
\end{align}
\end{proposition}

%
Note that there is no general answers to the question whether $\psi$ can be convert to $\phi$ based on limited information gain from the pre-testing measurement; that is, there are three possible answers for a one-shot state transformation question: ``yes'', ``no'' and ``lack of information''.

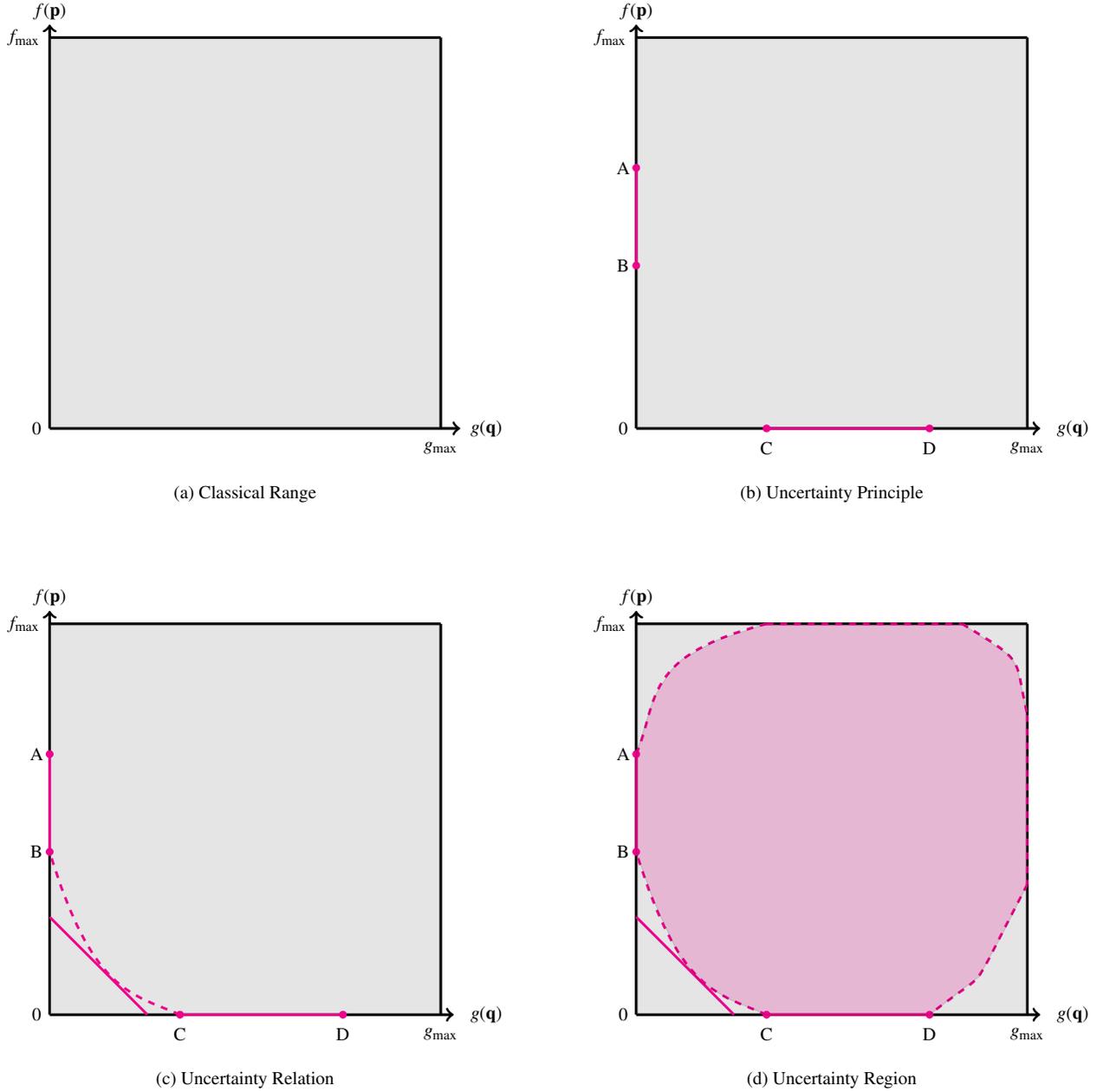
\begin{figure}[h]
\centering
\begin{tikzpicture}
    
\draw[very thick,->] (-9,-3) -- (-2.7,-3);
\draw[very thick,->] (-9,-3) -- (-9,3.2);
\draw[very thick] (-9,3) -- (-3,3);
\draw[very thick] (-3,-3) -- (-3,3);

\draw[fill=gray,opacity=0.2] (-9,-3) -- (-3,-3) -- (-3,3) -- (-9,3);

 \node at (-9,3.4) {\footnotesize $f(\p)$};
 \node at (-2.3,-3) {\footnotesize $g(\q)$};
 \node at (-9.2,-3) {\footnotesize $0$};
 
  \node at (-9.4,3) {\footnotesize $f_{\max}$};
  \node at (-3,-3.3) {\footnotesize $g_{\max}$};
  \node at (-6,-4) {\footnotesize (a) Classical Range};
  
\draw[very thick,->] (0,-3) -- (6.2,-3);
\draw[very thick,->] (0,-3) -- (0,3.2);
\draw[very thick] (0,3) -- (6,3);
\draw[very thick] (6,-3) -- (6,3); 

\draw[fill=gray,opacity=0.2] (0,-3) -- (6,-3) -- (6,3) -- (0,3);

  \node at (-0.4,3) {\footnotesize $f_{\max}$};
  \node at (-0.2,1) {\footnotesize A};
  \node at (-0.2,-0.5) {\footnotesize B};
  \node at (6,-3.3) {\footnotesize $g_{\max}$};
  \node at (2,-3.3) {\footnotesize C};
  \node at (4.5,-3.3) {\footnotesize D};
  \node at (0,3.4) {\footnotesize $f(\p)$};
 \node at (6.7,-3) {\footnotesize $g(\q)$};
  \node at (-0.2,-3) {\footnotesize $0$};
  \node at (3,-4) {\footnotesize (b) Uncertainty Principle};
  
\draw[line width=0.4mm,magenta] (0,-0.5) -- (0,1);
\draw[line width=0.4mm,magenta] (4.5,-3) -- (2,-3);
\filldraw [magenta] (0,-0.5) circle [radius=1.5pt];
\filldraw [magenta] (0,1) circle [radius=1.5pt];
\filldraw [magenta] (4.5,-3) circle [radius=1.5pt];
\filldraw [magenta] (2,-3) circle [radius=1.5pt];

\draw[very thick,->] (-9,-12) -- (-2.8,-12);
\draw[very thick,->] (-9,-12) -- (-9,-5.8);
\draw[very thick] (-9,-6) -- (-3,-6);
\draw[very thick] (-3,-12) -- (-3,-6); 

\draw[fill=gray,opacity=0.2] (-9,-12) -- (-3,-12) -- (-3,-6) -- (-9,-6);

  \node at (-9.4,-6) {\footnotesize $f_{\max}$};
  \node at (-9.2,-8) {\footnotesize A};
  \node at (-9.2,-9.5) {\footnotesize B};
  \node at (-3,-12.3) {\footnotesize $g_{\max}$};
  \node at (-7,-12.3) {\footnotesize C};
  \node at (-4.5,-12.3) {\footnotesize D};
  \node at (-9,-5.6) {\footnotesize $f(\p)$};
 \node at (-2.3,-12) {\footnotesize $g(\q)$};
  \node at (-9.2,-12) {\footnotesize $0$};
  \node at (-6,-13) {\footnotesize (c) Uncertainty Relation};
  
\draw[line width=0.4mm,magenta] (-9,-9.5) -- (-9,-8);
\draw[line width=0.4mm,magenta] (-4.5,-12) -- (-7,-12);
\filldraw [magenta] (-9,-9.5) circle [radius=1.5pt];
\filldraw [magenta] (-9,-8) circle [radius=1.5pt];
\filldraw [magenta] (-4.5,-12) circle [radius=1.5pt];
\filldraw [magenta] (-7,-12) circle [radius=1.5pt];

\draw[line width=0.4mm,magenta] (-7.5,-12) -- (-9,-10.5);
\draw[line width=0.4mm,magenta,dashed] (-7,-12) .. controls (-7.75,-11.675) and (-8.35,-11.65) .. (-9,-9.5);

\draw[very thick,->] (0,-12) -- (6.2,-12);
\draw[very thick,->] (0,-12) -- (0,-5.8);
\draw[very thick] (0,-6) -- (6,-6);
\draw[very thick] (6,-12) -- (6,-6); 

\draw[fill=gray,opacity=0.2] (0,-12) -- (6,-12) -- (6,-6) -- (0,-6);

  \node at (-0.4,-6) {\footnotesize $f_{\max}$};
  \node at (-0.2,-8) {\footnotesize A};
  \node at (-0.2,-9.5) {\footnotesize B};
  \node at (6,-12.3) {\footnotesize $g_{\max}$};
  \node at (2,-12.3) {\footnotesize C};
  \node at (4.5,-12.3) {\footnotesize D};
  \node at (0,-5.6) {\footnotesize $f(\p)$};
 \node at (6.7,-12) {\footnotesize $g(\q)$};
  \node at (-0.2,-12) {\footnotesize $0$};
  \node at (3,-13) {\footnotesize (d) Uncertainty Region};
  
\draw[line width=0.4mm,magenta] (0,-9.5) -- (0,-8);
\draw[line width=0.4mm,magenta] (4.5,-12) -- (2,-12);
\filldraw [magenta] (0,-9.5) circle [radius=1.5pt];
\filldraw [magenta] (0,-8) circle [radius=1.5pt];
\filldraw [magenta] (4.5,-12) circle [radius=1.5pt];
\filldraw [magenta] (2,-12) circle [radius=1.5pt];

\draw[line width=0.4mm,magenta] (1.5,-12) -- (0,-10.5);
\draw[line width=0.4mm,magenta,dashed] (2,-12) .. controls (1.25,-11.675) and (0.65,-11.65) .. (0,-9.5);

\draw[line width=0.4mm,magenta,dashed] (0,-8) .. controls (0.35,-7.2) and (0.05,-6.5) .. (2,-6);
\draw[line width=0.4mm,magenta,dashed] (2,-6) -- (5,-6);
\draw[line width=0.4mm,magenta,dashed] (5,-6) .. controls (6,-6.6) and (5.8,-6.5) .. (6,-7.4);
\draw[line width=0.4mm,magenta,dashed] (6,-7.4) -- (6,-10);
\draw[line width=0.4mm,magenta,dashed] (6,-10) .. controls (5,-11.9) and (5.5,-11.2) .. (4.5,-12);

\draw[fill=magenta,opacity=0.2] (0,-9.5) -- (0,-8) -- (0.25,-7.2) -- (0.35,-7) -- (0.55,-6.7) -- (0.7,-6.55) -- (1,-6.38) -- (1.25,-6.25) -- (2,-6) -- (5,-6) -- (5.75,-6.5) -- (5.85,-6.7) -- (6,-7.4) -- (6,-10) -- (5.65,-10.7) -- (5.32,-11.3) -- (5.29,-11.39) -- (4.5,-12) -- (2,-12) -- (1.4,-11.75) -- (1,-11.5) -- (0.665,-11.1) -- (0.38,-10.5) -- (0.25,-10.2) -- (0,-9.5) -- (0,-9.5);
  
\end{tikzpicture}
\caption{(color online) A schematic diagram depicts the classical range of $[0, f_{\max}] \times [0, g_{\max}]$, Heisenberg uncertainty principle, uncertainty relation $f(\p)+g(\q)\geqslant b$, and uncertainty region $\mathcal R(f,g)$.}
\label{classical}
\end{figure}

\section{Visualization of Uncertainty Principle, Uncertainty Relations, and Uncertainty Regions}

So far we have mostly talked about the rigorous mathematical formulation of uncertainty and complementarity, however a visualization that can help the reader understand the connections among them is also preferred. To start with, let us outline the trade-off between measurements in classical world. We still consider a protocol of black box testings, including pre-testing and post-testing, collect their probability distribution into vectors $\p$ and $\q$. Therefore, for given non-negative Schur concave functions $f$ and $g$, the uncertainty of $\p$ and $\q$ are bounded as 
\begin{align}
&f_{\max}:=f(1/n,1/n, \ldots, 1/n)\geqslant f(\p) \geqslant 0,\notag\\
&g_{\max}:=g(1/n,1/n, \ldots, 1/n)\geqslant g(\q) \geqslant 0.
\end{align}
Thus the classical region is the Cartesian product of $[0, f_{\max}]$ and $[0, g_{\max}]$, i.e. in classical world we have $\mathcal R(f,g) = [0, f_{\max}] \times [0, g_{\max}]$ with $\times$ stands for Cartesian product, as shown in Fig. \ref{classical} (a).

In quantum mechanics Heisenberg uncertainty principle refers to the situation, where for two incompatible measurements $M$ and $N$ there exists no sharp value for both of them. Equivalently, when the measurement outcome of $M$ appears with certainty, i.e. $f(\p)=0$, then the measurement outcome of $N$ becomes indeterministic. That means the range of $g(\q)$, denoted by $[C,D]$, does not contain the original point $o:=(0,0)$. By exchanging their roles, we have $0 \notin [B,A]$ with $[B,A]$ stands for the range of $f(\p)$ when $g(\q)=0$. The above discussion shows that uncertainty principle is equivalent to say $o \notin [B,A]$ and $o \notin [C,D]$, thus we obtained a picture illustrated by Fig. \ref{classical} (b). One can also ask, how about uncertainty relations? Let us emphasize that uncertainty relation $f(\p)+g(\q)\geqslant \min_{\rho}\{f(\p)+g(\q)\}:=b$ is tangent to the connection line of points $B$ and $C$ in Fig. \ref{classical} (c). Finally, we can delineate the boundary of uncertainty region $\mathcal R(f,g)$ in Fig. \ref{classical} (d). As argued in the main text, uncertainty region as illustrated in Fig. \ref{classical} (d) in more informative than both uncertainty principle and uncertainty relation of form $f(\p)+g(\q)\geqslant b$.

An interesting questions about uncertainty region is whether we can study the boundary of uncertainty region through all the tangent lines of $\mathcal R(f,g)$. Unfortunately, the answer is negative since the uncertainty region may not be a convex set in general. Explicit examples are given in Fig. \ref{MU} (b) and (c) in this  Supplemental Material. Therefore this is another reason why the considerations of the statistics set $\mathcal R := {\BigCup}_\rho\, \big\{ (\p(\rho),\q(\rho))\big\}$ in $\mathbb R^n \times \mathbb R^n$ and our complementary information principle are needed.

\end{document}